\documentclass{amsart}

\usepackage{amsmath}
\usepackage{hyperref}
\usepackage{graphicx}
\allowdisplaybreaks

\newcommand{\intd}{\, \textnormal{d}}
\newcommand{\ud}{\textnormal{d}}
\newcommand{\Tr}{\textnormal{Tr}}
\newcommand{\sgn}{\textnormal{sign}}
\newcommand{\mb}[1]{\mathbf{#1}}
\newcommand{\mbb}[1]{\mathbb{#1}}

\newcommand{\mcl}[1]{\mathcal{#1}}
\newcommand{\mnm}[1]{\textnormal{#1}}
\newcommand{\mfk}[1]{\mathfrak{#1}}
\newcommand{\balpha}{\boldsymbol{\alpha}}
\newcommand{\bxi}{\boldsymbol{\xi}}
\newcommand{\bzeta}{\boldsymbol{\zeta}}
\newcommand{\bx}{\mathbf{x}}
\newcommand{\by}{\mathbf{y}}
\newcommand{\bz}{\mathbf{z}}
\newcommand{\bX}{\mathbf{X}}

\newtheorem{theorem}{Theorem}
\newtheorem{lemma}[theorem]{Lemma}

\newtheorem{definition}[theorem]{Definition}

\begin{document}

\title[Dunkl jump processes]{Dunkl jump processes:\\relaxation and a phase transition}
\author{Sergio Andraus}
\address{Department of Physics, Faculty of Science and Engineering, Chuo University, Kasuga 1-13-27, Bunkyo-ku, Tokyo 112-8551, Japan}
\email[S. Andraus]{andraus@phys.chuo-u.ac.jp}
\date{\today}

\begin{abstract}
Dunkl processes are multidimensional Markov processes defined through the use of Dunkl operators. Their paths show discontinuities, and so they can be separated into their continuous (radial) part, and their discontinuous (jump) part. While radial Dunkl processes have been studied thoroughly due to their relationship with families of stochastic particle systems such as the Dyson model and Wishart-Laguerre processes, Dunkl jump processes have gone largely unnoticed after the initial work of Gallardo, Yor and Chybiryakov. We study the dynamical properties of the latter processes, and we derive their master equation. By calculating the asymptotic behavior of their total jump rate, we find that the jump processes of types $A_{N-1}$ and $B_N$ undergo a phase transition when the parameter $\beta$ decreases toward one in the bulk scaling limit. In addition, we show that the relaxation behavior of these processes is given by a non-trivial power law, and we derive an asymptotic relation for the relaxation exponent in order to discuss its $\beta$-dependence.
\end{abstract}

\maketitle

\section{Introduction and Main Results}

Dunkl processes \cite{RoslerVoit98, RoslerVoit08} are a family of multidimensional stochastic processes defined as a generalization of Brownian motion using Dunkl operators \cite{Dunkl89}. The latter are differential-difference operators which depend on the choice of a root system $R$, which is a finite set of vectors which generates reflection groups, and a set of parameters, called multiplicities. In order to illustrate the main characteristics of a Dunkl process, we take the root system of type $A_{N-1}$ as an example. Consider a group of $N$ Brownian particles on the real line; then, every configuration can be represented as a vector in $\mbb{R}^N$. If we denote the transition probability density (TPD) of the process going from the configuration $\bx\in\mbb{R}^N$ to the configuration $\by\in\mbb{R}^N$ in a time $t>0$ by $p_{A}(t,\by|\bx)$, then the backward Fokker-Planck equation of the Dunkl process reads
\begin{align}
\frac{\partial}{\partial t}p_{A}(t,\by|\bx)&=\frac{1}{2}\sum_{i=1}^N\frac{\partial^2}{\partial x^2}p_{A}(t,\by|\bx)\notag\\
&\quad+\frac{\beta}{2}\sum_{1\leq i\neq j\leq N}\Big[\frac{1}{x_i-x_j}\frac{\partial}{\partial x_i}-\frac{1-\sigma_{ij}}{2(x_i-x_j)^2}\Big]p_{A}(t,\by|\bx).\label{eq:DunklATPD}
\end{align}
The operator $\sigma_{ij}$ permutes the $i$-th and $j$-th components of the vector $\bx$, and $\beta/2>0$ is the sole multiplicity in this case. There are two main things to note about this process: the first is that, if the difference term (second term in brackets on the rhs) were not a part of this equation, the process would be equivalent to the Dyson Brownian motion model \cite{Dyson62} with parameter $\beta$. This fact is well-known \cite{DemniBook}, and in general for every Dunkl process there is a continuous version of it, called a \emph{radial} Dunkl process \cite{GallardoYor05}. The second is that the difference term introduces discontinuities to the process. This paper is focused on the discontinuous part, the Dunkl jump process.

\begin{figure}
\[
\begin{array}{cc}
\includegraphics[width=0.55\textwidth]{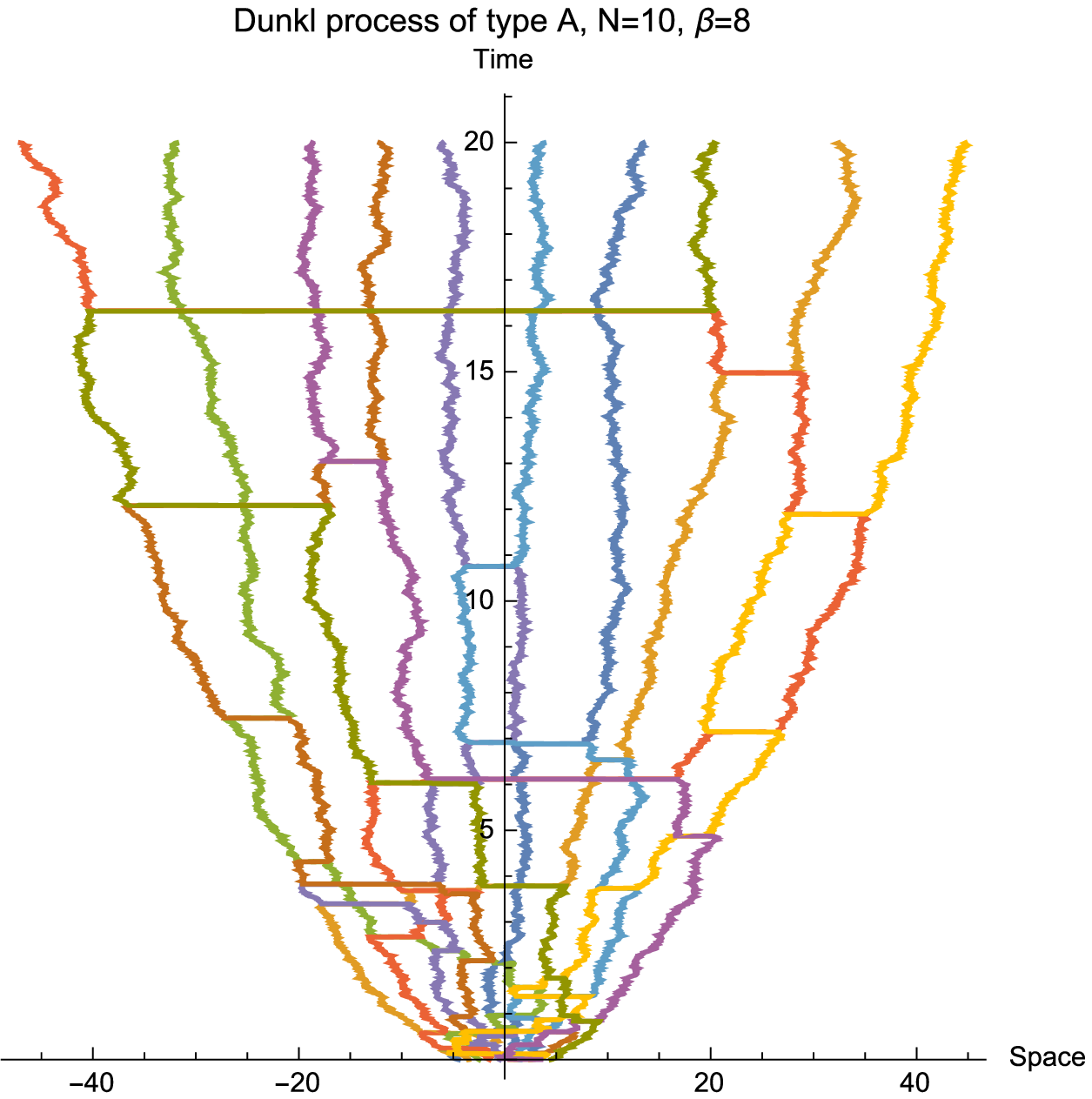}&\includegraphics[width=0.331\textwidth]{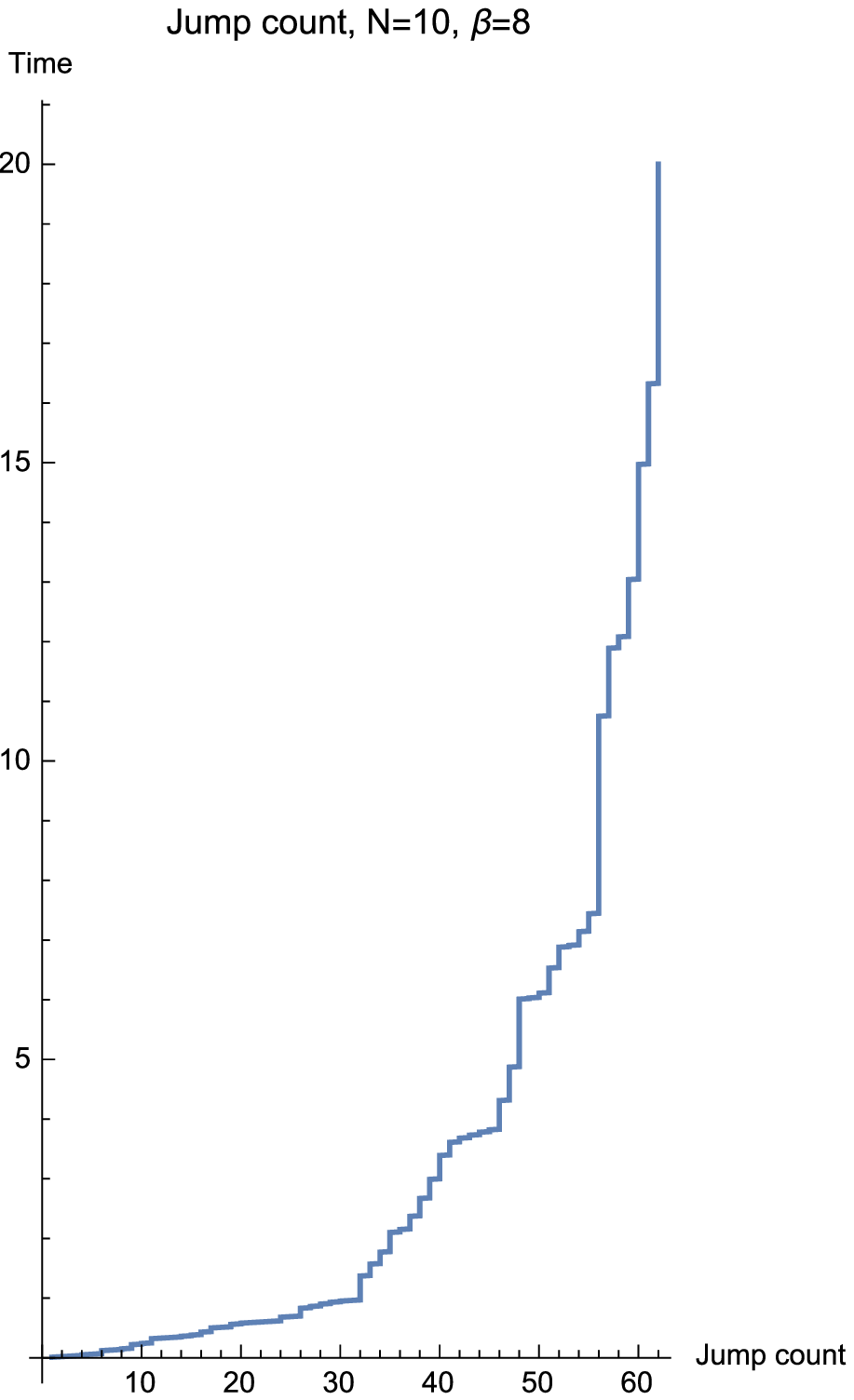}
\end{array}
\]
\caption{\label{fig:sampleDunklA}Sample of the Dunkl process of type $A_{N-1}$ and its jump count for $N=10$, $\beta=8$, and a centered, equally spaced initial configuration with unit distance between nearest neighbors. The horizontal lines represent jumps.}
\end{figure}

We plot a sample of the type-$A_{N-1}$ Dunkl process in Figure~\ref{fig:sampleDunklA}. The left plot shows the path that each particle takes. As the particles diffuse and repel each other due to the first term in brackets on the rhs of \eqref{eq:DunklATPD}, they exchange positions randomly. These exchanges are represented by horizontal lines in the plot, and it is apparent that the probability of an exchange grows the closer two particles are to each other. Then, it becomes clear that this type of Dunkl process is composed of two parts: the diffusing path itself, and the particle order. Since the exchange interactions generate all possible permutations of $N$ objects, it follows that the discontinuous part is a stochastic process that takes values in the symmetric group $S_N$.

The first study of the jumps in Dunkl processes was carried out by Gallardo and Yor in \cite{GallardoYor06A}, where they succeeded in deriving a skew-product decomposition of the Dunkl process of rank one in terms of its jump part and its continuous part. They also succeeded in expressing the multidimensional case in terms of two martingale parts, one continuous and the other completely discontinuous, in \cite{GallardoYor06B}. Later, the first skew-product decomposition in the multidimensional case was achieved by Chybiryakov in \cite{Chybiryakov08}. However, there have been no studies on the Dunkl jump processes since, and very little is known about these jump processes, particularly in physical terms. 

This does not mean that Dunkl processes and Dunkl operators are unknown in physics. In fact, Dunkl operator theory has been applied to great effect in the area of integrable systems \cite{GenestVinetZhedanov14,deBieGenestVinet16}, particularly those of Calogero-Moser-Sutherland type \cite{BakerForrester, Khastgir, EtingofGinzburg}. In addition, the radial Dunkl processes of types $A_{N-1}$ and $B_N$ are equivalent to the Dyson model \cite{Dyson62} and to the Wishart-Laguerre models \cite{Bru,KonigOConnell}, which are well-known random matrix eigenvalue processes \cite{Forrester}. However, the exchange mechanism responsible for Dunkl jump processes has only rarely been studied in physics, mostly in the form of a generalized spin interaction \cite{Polychronakos92, Frahm93, HikamiWadati93}.

In the present paper, we study the dynamics of Dunkl jump processes by making use of the observation that a Dunkl process' trajectory between any two jumps is given by the stochastic differential equation (SDE) of its radial part. Then, it is clear that the jumps depend on the radial process' trajectory but not vice versa, so in order to treat the jump process in isolation we express the rate of every possible jump as an expectation taken with respect to the law of the radial process. Then, we derive the time-evolution equation for the jump process' law (namely, the master equation), and obtain our main results from it.

For the first result, we consider processes of type $A_{N-1}$ and $B_N$, with $N$ particles starting from an initial configuration $\bx_0$ in the closed Weyl chamber $C_{W_R^N}$, where \begin{align}
C_{W_A^N}&:=\{\bx\in\mbb{R}^N:x_1\leq x_2\leq\cdots\leq x_N \},\notag\\
C_{W_B^N}&:=\{\bx\in\mbb{R}^N:0\leq x_1\leq x_2\leq\cdots\leq x_N \},
\end{align} and the processes evolve for a time $t$. In such a setting, we denote the total jump rate by $\Lambda_\beta^{(R)}(t|\bx_0)$ with $R=A_{N-1}$ or $B_N$, and we find the following behavior in the \emph{bulk scaling limit} where $t=N$ and $N\to\infty$ \cite{KatoriTanemura07}. \begin{theorem}\label{th:phasetransition}
	Set $\beta>1$ and consider a sequence of initial configurations $\{\bx_0^{(N)}\in C_{W_R^N}\}_{N=2}^\infty$ such that $\|\bx_{0}^{(N)}\|\leq K$ with $K>0$ fixed. In the bulk scaling limit, the total jump rate per particle for Dunkl jump processes of type $R=A_{N-1}$ and $B_N$ with equal multiplicities $\beta/2$ for all roots is given by
	\begin{equation}
	\lim_{N\to\infty}\frac{1}{N}\Lambda_\beta^{(R)}(N|\bx_0^{(N)})=\frac{\beta}{c_R(\beta-1)},
	\end{equation}
	with $c_{A_{N-1}}=8$ and $c_{B_N}=4$, respectively. Therefore, we have a phase transition at $\beta_c=1$ with critical exponent equal to one.
\end{theorem}
This phase transition is closely related to whether particles in the underlying radial processes collide. It is well-known \cite{Demni08} that whenever all of the multiplicities are larger than or equal to 1/2, the first collision time tends to infinity almost surely; here, all multiplicities are equal to $\beta/2$, meaning that there are no collisions for $\beta\geq1$. This indicates that the phase transition lies between a non-colliding (ordered) phase and a colliding (disordered) phase.

The second result involves the main dynamical properties of the jumps, as derived from the master equation. In the example depicted in Figure~\ref{fig:sampleDunklA}, we can observe that every time a jump occurs only two particles exchange their positions, so there exist $N(N-1)/2$ different kinds of jumps, and every kind of jump has its own rate. In general we denote the jump rates by $\lambda_{\beta}(t,\balpha|\bx_{0})$, where $\balpha$ is a root in $R$, and at the same time, $\balpha$ indicates the type of jump. In the case $A_{N-1}$, $\balpha$ can be simply replaced by two indices $i,j\in 1,\ldots,N$ referring to the particles exchanging positions. We also denote the sum of all possible jumps by $\Lambda_\beta(t|\bx_0)$. Then, the jump process takes values in the reflection (Weyl) group $W$ generated by reflection operators $\sigma_{\balpha}$ along $\balpha\in R$. We give the master equation in the following statement.
\begin{lemma}\label{lm:masterequation}
	Denote by $P^{\mcl{J}}_{\beta}(t,\tau|\bx_0)$ the probability of finding the jump process associated to a Dunkl process started from $\bx_0$ at $\tau\in W$ after a time $t$. Then, for $\beta>1$, $P^{\mcl{J}}_{\beta}(t,\tau|\bx_0)$ obeys the master equation
	\begin{equation}\label{eq:master}
	\frac{\partial}{\partial t}P^{\mcl{J}}_{\beta}(t,\tau|\bx_0)=\sum_{\balpha\in R_+}\lambda_\beta(t,\balpha|\bx_0)P^{\mcl{J}}_{\beta}(t,\tau\sigma_{\balpha}|\bx_0)-\Lambda_\beta(t|\bx_0)P^{\mcl{J}}_{\beta}(t,\tau|\bx_0).
	\end{equation}
\end{lemma}
From results in \cite{AndrausMiyashita15}, we deduce that the asymptotic behavior of all jump rates is given by
\begin{equation}\label{eq:jumprateslarget}
\lambda_{\beta}(t,\balpha|\bx_{0})=\frac{1}{t}\lambda_{\beta}(1,\balpha|\mb{0})+O(\|\bx_0\|^2t^{-2}),
\end{equation}
where $\bx=\mb{0}$ means that $x_i=0\ {}^\forall i=1,\ldots N$. After inserting this relation in the master equation, we immediately obtain the following.
\begin{theorem}\label{th:relaxation}
	For $\bx_0\neq\mb{0}$ and $\beta >1$, relaxation asymptotics in Dunkl jump processes follow a power law with an exponent less than or equal to unity in absolute value.
\end{theorem}
The meaning of this statement in the case $A_{N-1}$ is that the equilibrium law for the jump process gives all elements of $S_N$ equal probability, and that for large times $t$ the probability $P^{\mcl{J}}_{\beta}(t,\tau|\bx_0)$ is given by a $\beta$-dependent relaxation exponent $r(\beta)<0$ as
$$
P^{\mcl{J}}_{\beta}(t,\tau|\bx_0)=\frac{1}{N!}+C(\tau)t^{r(\beta)}+o(t^{r(\beta)})
$$
for $C(\tau)$ a static function of $\tau$. The $\beta$-dependence of $r(\beta)$ is a more delicate matter because of \eqref{eq:jumprateslarget}, as the two terms on the rhs are representative of the two factors at play in the relaxation process. When $\bx_0=\mb{0}$ only the first term survives, and the dynamical part of $P^{\mcl{J}}_{\beta}(t,\tau|\mb{0})$ is given by the solution to a simple matrix eigenvalue problem. In this case the relaxation exponent is the eigenvalue closest to zero, which we denote $r_1(\beta)<0$. If, on the other hand, $\bx_0\neq\mb{0}$, then the second term in \eqref{eq:jumprateslarget} becomes important, as it automatically 
%introduces a perturbation of order $t^{-2}$ into the master equation, and 
adds a perturbation of order $t^{-1}$ to the solution for $\bx_0=\mb{0}$. Then, if $r_1(\beta)<-1$, the term of order $t^{-1}$ is dominant: this is the basis for Theorem~\ref{th:relaxation}.

Then, it becomes important to find the relaxation exponent for the case $\bx_0=\mb{0}$ as a function of $\beta$ to clarify whether the initial configuration $\bx_0$ dominates the relaxation dynamics. This is the motivation for our final result.
\begin{theorem}\label{th:perturbation}
	For every $M>0$, there exists $C(M)>1$ such that, for $\beta>C(M)$ and equal multiplicities $\beta/2$ for all roots, the dynamical exponents $\{r_i(\beta)\}$ (with $r_{i+1}(\beta)<r_i(\beta)$ for every $i>0$) of a Dunkl jump process with $\bx_0=\mb{0}$ obey the first-order asymptotics
	\begin{equation}
	r_i(\beta)=-\Big(|r_i^{(0)}|+\frac{|r_i^{(1)}|}{\beta}\Big)+O(\beta^{-2})+O(\mnm{e}^{-M}),
	\end{equation}
	with $r_i^{(0)}$, $r_i^{(1)}$ constant with respect to $\beta$.
\end{theorem}
This behavior of $r_i(\beta)$ indicates that there exists a large-$\beta$ regime in which the relaxation exponent $r_1(\beta)$ decreases in magnitude with $\beta$, and the effect of $\bx_0$ at large $t$ depends ultimately on whether $|r_1^{(0)}|<1$. We finish by considering the case $A_{N-1}$, and we find that, because $|r_1^{(0)}|=1/2$, the relaxation exponent tends toward $1/2$ when $\beta\geq 1$ is large, but becomes $1$ whenever $|r_1(\beta)|>1$.

This paper is organized as follows. In Section~\ref{sec:definitions}, we give a brief overview of Dunkl processes and define the quantities involved in our results. In Section~\ref{sec:jumpprocess}, we define Dunkl jump processes, derive their master equation, and discuss several immediate properties of the latter.
In Section~\ref{sec:phasetransition}, we discuss the jump counting process, and calculate the asymptotic form of the total jump rate. From there, we prove Theorem~\ref{th:phasetransition}. 
In Section~\ref{sec:relaxation}, we consider the relaxation behavior of the jump process. The asymptotics from Section~\ref{sec:phasetransition} become instrumental in the calculations involved. We prove Theorem~\ref{th:relaxation} by calculating the order of magnitude of the difference between solutions of the master equation for $\bx_0\neq \mb{0}$ and $\bx_0=\mb{0}$. We also prove Theorem~\ref{th:perturbation} by considering the case $\bx_0=\mb{0}$ in the limit $\beta\to\infty$ and performing a first-order perturbation on its solution. 
In Section~\ref{sec:relaxationA}, we consider the case $A_{N-1}$ explicitly, and make use of its connection to Polychronakos-Frahm (PF) spin chains \cite{Polychronakos92, Frahm93} to find the relaxation exponent when $\bx_0=\mb{0}$ and $\beta\to\infty$, which is exactly one half. Then, it becomes clear that in the case $A_{N-1}$ there exists a relaxation regime that is not dominated by the initial configuration at large values of $\beta$.
Finally, we give our concluding remarks and briefly state several related open problems in Section~\ref{sec:remarks}.

\section{Setting and definitions}\label{sec:definitions}

We briefly fix the definitions necessary for our results starting from Dunkl operators \cite{Dunkl89}; we encourage the reader to consult \cite{RoslerVoit08,DunklXu} for a more detailed treatment of the topics covered here.
Consider the  $N$-dimensional Euclidean space $\mbb{R}^N$ and the reflection operator $\sigma_{\balpha}\in O(N)$ defined by
\begin{equation}
\sigma_{\balpha}\bx:=\bx-2\frac{\balpha\cdot\bx}{\|\balpha\|^2}\balpha,
\end{equation}
with the Euclidean inner product between $\balpha$ and $\bx\in\mbb{R}^N$ denoted by $\balpha\cdot\bx$,
and $\|\balpha\|^2=\balpha\cdot\balpha$.

Let us fix a root system $R\subset\mbb{R}^N\setminus\{0\}$, namely, a finite set of vectors, or \emph{roots}, 
invariant under the action of the reflections along its own elements.
For simplicity, we assume $R$ to be reduced, that is, if $\balpha,\ \bxi\in R$ 
and there exists a constant $c$ such that $\balpha=c\bxi$, then $c=\pm1$.
In fact, if $\balpha\in R$, then $-\balpha\in R$, so we only require one half of $R$ to obtain all the reflections $R$ can generate.
Thus, one can choose an arbitrary vector, say $\mb{m}\in\mbb{R}^N$, such that $\mb{m}\cdot\balpha\neq 0$ for every root $\balpha$,
and with that, we can define the positive subsystem $R_+$ as the subset of $R$ which contains all the roots that satisfy $\mb{m}\cdot\balpha>0$.
We also define by $C_W=\{\bx\in\mbb{R}^N: \balpha\cdot\bx\geq 0, {}^\forall \balpha\in R_+\}$ the closed cone (or Weyl chamber) which contains the vector $\mb{m}$.

$R$ defines a reflection group, which we denote by $W\subset O(N)$, 
and $R$ can be partitioned by the disjoint orbits of the roots under the action of $W$.
One may then assign a parameter, called a \emph{multiplicity}, to each disjoint orbit, defining what is called a multiplicity function.
There is no restriction on these multiplicities, but for the setting we consider here we will make the following assumptions.
Define a positive real parameter $\beta$ and a multiplicity function $k(\balpha)\geq 0$, with the requirement that one of the multiplicities be equal to one. We also introduce the sum of multiplicities
\begin{equation}\label{eq:gammadefinition}
\gamma:=\sum_{\balpha\in R_+}k(\balpha)
\end{equation}
for use in later sections.

Then, the Dunkl operators $\{T_i\}_{i=1}^N$ are defined for functions $f\in \mcl{C}^1(\mbb{R}^N)$ by
\begin{equation}
T_if(\bx):=\frac{\partial}{\partial x_i}f(\bx)+\frac{\beta}{2}\sum_{\balpha\in R_+}k(\balpha)\frac{f(\bx)-f(\sigma_{\balpha}\bx)}{\balpha\cdot\bx}\alpha_i.
\end{equation}
We choose our multiplicities in this manner so that we can treat the parameter $\beta$ in the same way as in the $\beta$-ensembles from random matrix theory \cite{DumitriuEdelman02}. Then, for example, the condition that all multiplicities be equal to $\beta/2$ is realized by setting $k(\balpha)\equiv 1$.
It is known \cite{DunklXu} that Dunkl operators satisfy some of the properties of partial derivatives.
In particular, they commute with each other, so it makes sense to generalize well-known differential operators using Dunkl operators.
One such generalization is the Dunkl Laplacian $\Delta_k:=\sum_{i=1}^NT_i^2$.
It was shown in \cite{Dunkl89, DunklXu} that its explicit form for $f\in \mcl{C}^2(\mbb{R}^N)$ is
\begin{equation}\label{eq:DunklLaplacian}
\Delta_k f(\bx)=\sum_{i=1}^N\frac{\partial^2}{\partial x^2}f(\bx)+\beta\sum_{\balpha\in R_+}k(\balpha)\Big[\frac{\balpha\cdot\nabla f(\bx)}{\balpha\cdot\bx}-\frac{\|\balpha\|^2}{2}\frac{f(\bx)-f(\sigma_{\balpha}\bx)}{(\balpha\cdot\bx)^2}\Big].
\end{equation}

Dunkl processes were defined in \cite{RoslerVoit98} as the left-limited, right-continuous adapted Markov processes in the probability space $(\Omega,\mcl{F},\mbb{P})$ whose infinitesimal generator is $\Delta_k/2$. Namely, if the TPD of a process going from $\bx$ to $\by\in\mbb{R}^N$ after a time $t>0$ is denoted by $p(t,\by|\bx)$, and it satisfies the backward Fokker-Planck equation
\begin{equation}\label{eq:BackFPE}
\frac{\partial}{\partial t}p(t,\by|\bx)=\frac{1}{2}\Delta_kp(t,\by|\bx),
\end{equation}
where $\Delta_k$ acts on $\bx$, then it is a Dunkl process. The explicit form of the TPD $p(t,\by|\bx)$ is given by
\begin{equation}\label{eq:TPD}
p(t,\by|\bx)=\frac{\text{e}^{-(\|\bx\|^2+\|\by\|^2)/2t}}{c_\beta t^{N/2}}E_\beta\Big(\frac{\bx}{\sqrt{t}},\frac{\by}{\sqrt{t}}\Big)w_\beta\Big(\frac{\by}{\sqrt{t}}\Big).
\end{equation}
Here, $c_\beta$ is a normalization constant given by the Selberg integral
\begin{equation}
c_\beta=\int_{\mbb{R}^N}\text{e}^{-\|\bx\|^2/2}w_\beta(\bx)\intd^N \bx,
\end{equation}
and the \emph{Dunkl kernel} $E_\beta(\bx,\by)$ is the generalization of the exponential function defined by the relations
\begin{equation}
E_\beta(\mb{0},\by)=E_\beta(\bx,\mb{0})=1\quad \text{and}\quad T_iE_\beta(\bx,\by)=y_iE_\beta(\bx,\by)\ {}^\forall i\in\{1,\ldots,N\}.
\end{equation}
It is important to stress that while the existence of the Dunkl kernel has been proved \cite{Dunkl91},
at present there is no known general explicit form for it, except in a few particular cases.
Finally, note that $p(t,\by|\bx)$ obeys the important scaling property
\begin{equation}\label{eq:TPDScaling}
p(t,\by|\bx)\intd^N \by=p\Big(1,\frac{\by}{\sqrt{t}}\Big|\frac{\bx}{\sqrt{t}}\Big)\frac{\ud^N\by}{t^{N/2}}.
\end{equation}

It is clear that Dunkl processes are discontinuous, as evidenced by the difference term in \eqref{eq:DunklLaplacian}. 
However, one can extract the continuous part of this process by considering the sum
\begin{equation}
\hat{p}(t,\by|\bx)=\sum_{\rho\in W}{p}(t,\by|\rho\bx).
\end{equation}
Inserting $\hat{p}(t,\by|\bx)$ into \eqref{eq:BackFPE} cancels out the difference term; 
remark that $\hat{p}(t,\by|\bx)$ is normalized in the Weyl chamber $C_W$.
Then, the processes defined by $\hat{p}(t,\by|\bx)$ and \eqref{eq:BackFPE} without the difference term are continuous; they are called \emph{radial} Dunkl processes \cite{GallardoYor05}, and their paths are contained in $C_W$.
Radial Dunkl processes have been studied thoroughly, and many of their properties have been elucidated in \cite{Demni08,Demni09}.
In particular, it is known that the stochastic differential equation for the radial Dunkl process $\hat{\bX}(t)$ is given by
\begin{equation}
\ud\hat{\bX}(t)=\ud \hat{\mb{B}}(t)-\frac{1}{2}\nabla\Phi(\hat{\bX}(t))\ud t,
\end{equation}
with $\hat{\mb{B}}(t)$ an $N$-dimensional Brownian motion,
\begin{equation}
\Phi(\bx):=-\log w_\beta(\bx),
\end{equation}
and
\begin{equation}
w_\beta(\bx):=\prod_{\balpha\in R_+}|\balpha\cdot\bx|^{\beta k(\balpha)}.
\end{equation}
It is also known that for $\beta\geq 1$ and $k(\balpha)\geq 1$ for all $\balpha\in R$, radial Dunkl processes do not hit the boundaries of $C_W$ almost surely, and that there exists a unique strong solution of this SDE whenever $\beta k(\balpha)>0$ for every root.

In the same way that one can isolate the continuous part of Dunkl processes, one can isolate the discontinuous part. The first steps in this direction were given in \cite{GallardoYor06A} for the rank-one case, and the corresponding multidimensional generalization was carried out in \cite{Chybiryakov08}. 

In this paper, however, we focus on the Dunkl jump process, which can be described as a continuous-time stochastic process on the group $W$, with the property that it becomes an inhomogeneous Poisson random walk when $\beta>1$. By considering the forward Fokker-Planck equation of the process,
\begin{eqnarray}
\frac{\partial}{\partial t}p(t,\by|\bx)&=&\frac{1}{2}\sum_{i=1}^N\frac{\partial^2}{\partial y^2}p(t,\by|\bx)-\frac{\beta}{2}\sum_{\balpha\in R_+}k(\balpha)\frac{\balpha\cdot\nabla p(t,\by|\bx)}{\balpha\cdot\by}\nonumber\\
&&\label{eq:ForwardFPE}+\frac{\beta}{2}\sum_{\balpha\in R_+}k(\balpha)\frac{\|\balpha\|^2}{2}\frac{p(t,\by|\bx)+p(t,\sigma_{\balpha}\by|\bx)}{(\balpha\cdot\by)^2},
\end{eqnarray}
one finds that the probability density of a jump from $\bx$ (such that $\bx\cdot\balpha\neq 0$) to a point $\by$ at a non-zero distance in an interval $\ud t$ is given by
\begin{equation}\label{eq:levykernel}
p(\ud t,\by|\bx)=\frac{\beta}{2}\sum_{\balpha\in R_+}k(\balpha)\frac{\|\balpha\|^2}{2}\frac{\delta(\sigma_{\balpha}\by-\bx)}{(\balpha\cdot\bx)^2}\ud t.
\end{equation}
This follows immediately from the L{\'e}vy measure of the process, as stated in \cite{GallardoYor06B,Meyer67}. From this expression it is clear that, if a jump occurs, the arrival point must be one of the reflections of the vector $\bx$ in the direction of one of the root vectors $\balpha$. Moreover, if we denote the Dunkl process by $\bX(t)$ and impose the initial condition $\bX(0)=\bx_0$, the probability rate that any jump occurs in the interval $[t,t+\ud t)$ is given by the integral
\begin{eqnarray}
\Lambda_\beta(t|\bx_0)&=&\int_{\mbb{R}^N}\int_{\mbb{R}^N}\frac{\beta}{2}\sum_{\balpha\in R_+}k(\balpha)\frac{\|\balpha\|^2}{2}\frac{\delta(\sigma_{\balpha}\by-\bx)}{(\balpha\cdot\bx)^2}\intd^N \by\ p(t,\bx|\bx_0)\intd^N \bx\nonumber\\
&=&\int_{C_W}\frac{\beta}{2}\sum_{\balpha\in R_+}k(\balpha)\frac{\|\balpha\|^2}{2}\frac{\hat{p}(t,\bx|\bx_0)}{(\balpha\cdot\bx)^2}\intd^N \bx>0.\label{eq:totalrate}
\end{eqnarray}
Note the change in notation: henceforth, $\bx_0$ will represent the initial configuration, and $\bx$ the configuration at time $t$. It was shown in \cite{GallardoYor06B,Lepingle12} that this integral converges whenever $\beta k(\balpha)>1$, and diverges otherwise. The reason for this is that, due to the presence of the weight function $w_\beta(\bx)$ in \eqref{eq:TPD}, the behavior of the singularity near $\balpha\cdot\bx=0$ goes like $|\balpha\cdot\bx|^{\beta k(\balpha)-2}$, so the singularity is integrable whenever $\beta k(\balpha)>1$ for all $\balpha$.

\section{The Dunkl jump process}\label{sec:jumpprocess}

In order to construct the jump process, we recall the basic properties of the jumps and we recast the main result from \cite{Chybiryakov08} (Theorem~19) in a more familiar form. Since we have established that the jump rate is finite when $\beta k(\balpha)>1$ for all $\balpha$, it is clear that there exist finite, non-zero time intervals between jumps in this regime. Using this fact, we can prove the following statement.

\begin{lemma}\label{jumplemma}
Suppose that $\beta k(\balpha)>1$ for every root $\balpha$. Then, for every Dunkl process $\bX(t)$ and its corresponding radial part $\hat{\bX}(t)$, there exists a Poisson random walk $\rho(t)$ on the Weyl group $W$ such that
\begin{equation}\label{eq:jumplemma}
\bX(t)=\rho(t)\hat{\bX}(t),
\end{equation}
where the equality holds in law. Moreover, if for every interval $[S,T)$ without jumps the Brownian motions $\mb{B}(t)$ and $\hat{\mb{B}}(t)$ which drive $\bX(t)$ and $\hat{\bX}(t)$ respectively are related by
\begin{equation}
\mb{B}(t)-\mb{B}(S)=\rho(t)(\hat{\mb{B}}(t)-\hat{\mb{B}}(S))
\end{equation}
for $t\in[S,T)$, then the equality holds pathwise.
\end{lemma}

\begin{proof}
Without loss of generality, we set $\bx_0\in C_W$, fix $\omega\in\Omega$, and proceed by induction, while keeping in mind that all objects in this derivation are functions of $\omega$. Clearly, $\bX(t)=\hat{\bX}(t)$ and $\rho(t)=\text{id}$ before the first jump, and $\hat{\bX}(t)$ is the unique strong solution if the radial SDE with respect to $\omega$. Now, suppose that after the $n$-th jump, \eqref{eq:jumplemma} holds, and denote the (random) time of said jump by $t_n$. The timing of the jumps is decided by an inhomogeneous Poisson process with rate function $\Lambda_\beta(t,\bx_0)$. The next jump occurs at time $t_{n+1}$, and we know by \eqref{eq:levykernel} that there exists a root $\balpha_{n+1}$ such that
\begin{equation}
\bX(t_{n+1})=\sigma_{\balpha_{n+1}}\bX(t_{n+1}^-)=\sigma_{\balpha_{n+1}}\rho(t_{n+1}^-)\hat{\bX}(t_{n+1}^-)=\sigma_{\balpha_{n+1}}\rho(t_{n+1}^-)\hat{\bX}(t_{n+1})
\end{equation}
by the induction hypothesis and the continuity of $\hat{\bX}(t)$. Then we can write for $t_{n+1}<t<t_{n+2}$,
\begin{align}
\bX(t)&=\sigma_{\balpha_{n+1}}\rho(t_{n+1}^-)\Big[\hat{\mb{B}}(t)+\frac{\beta}{2}\int_0^t\sum_{\bxi\in R_+}k(\bxi)\frac{\bxi}{\bxi\cdot\hat{\mb{X}}(s)}\intd s\Big]\notag\\
&=\sigma_{\balpha_{n+1}}\rho(t_{n+1}^-)\hat{\mb{B}}(t)+\frac{\beta}{2}\int_0^t\sum_{\bxi\in R_+}k(\bxi)\frac{\sigma_{\balpha_{n+1}}\rho(t_{n+1}^-)\bxi}{\bxi\cdot\hat{\mb{X}}(s)}\intd s\notag\\
&=\sigma_{\balpha_{n+1}}\rho(t_{n+1}^-)\hat{\mb{B}}(t)+\frac{\beta}{2}\int_0^t\sum_{\bxi\in R_+}k(\bxi)\frac{\bxi}{\bxi\cdot\mb{X}(s)}\intd s.
\end{align}
For the last line, we have used the substitution $\bxi^\prime=\sigma_{\balpha_{n+1}}\rho(t_{n+1}^-)\bxi$. Then it suffices to set $\rho(t)=\sigma_{\balpha_{n+1}}\rho(t_{n+1}^-)\in W$ and $\mb{B}(t)=\sigma_{\balpha_{n+1}}\rho(t_{n+1}^-)\hat{\mb{B}}(t)=\rho(t)\hat{\mb{B}}(t)$ for $t_{n+1}<t<t_{n+2}$ to obtain the result.
\end{proof}

\begin{definition}
We define the Dunkl jump process as the Poisson random walk $\rho(t)$ in Lemma~\ref{jumplemma}.
\end{definition}

Since the random timing of each jump is given by the jump rate $\Lambda_\beta(t,\bx_0)$, it only remains to calculate the rate corresponding to each of the possible jumps. From \eqref{eq:levykernel}, we see that the probability rate of a jump along the root $\balpha$ given $\hat{\bX}(t)=\bx$ and $\rho(t)=\tau$ (that is, $\bX(t)=\tau\bx$) with $\tau\in W$, $\bx\in C_W$ is given by
\begin{equation}
\frac{\beta\|\balpha\|^2}{4}\frac{k(\balpha)}{(\balpha\cdot\tau\bx)^2}=\frac{\beta\|\balpha\|^2}{4}\frac{k(\balpha)}{(\tau^{-1}\balpha\cdot\bx)^2},
\end{equation}
so by using Lemma~\ref{jumplemma} we obtain the rate function
\begin{equation}\label{eq:jumprates}
\lambda_\beta(t,\tau^{-1}\balpha|\bx_0)=\frac{\beta\|\balpha\|^2}{4}\int_{C_W}\frac{k(\balpha)}{(\tau^{-1}\balpha\cdot\bx)^2}\hat{p}(t,\bx|\bx_0)\intd^N \bx>0.
\end{equation}
This is, then, the probability rate that the jump process makes a transition where $\rho(t)$ goes from $\tau$ to $\sigma_{\balpha} \tau$. Note that
\begin{equation}\label{eq:transitionratecoefficientssum}
\sum_{\balpha\in R+}\lambda_\beta(t,\balpha|\bx_0)=\Lambda_\beta(t|\bx_0),
\end{equation}
so all of these rate functions converge whenever $\beta k(\balpha)>1$ for all $\balpha \in R$ by \cite{GallardoYor06B}. We can now derive the master equation by defining the function $P^{\mcl{J}}_{\beta}(t,\tau|\bx_0):[0,\infty)\times W\times C_W\to\mbb{R}_{\geq 0}$ as the probability of finding the jump process at the reflection group element $\tau\in W$ at time $t\in[0,\infty)$ when the radial process starts from $\bx_0\in C_W$, with the constraint that $\sum_{\tau\in W}P^{\mcl{J}}_{\beta}(t,\tau|\bx_0)=1$. 

\begin{proof}[Proof of Lemma~\ref{lm:masterequation}]
Partition the interval $[0,t)$ by setting $t_i=it/m$, $i=0,\ldots,m$, $\Delta t=t/m$ and set the notation $\sigma_i=\sigma_{\balpha_i}$. Finally, set $\tau_i=\sigma_i \sigma_{i-1}\cdots\sigma_1$, with $\tau_0=\text{id}$. Then, the probability that $\rho(t)$ performs exactly $n$ jumps along the roots $\balpha_1,\ldots,\balpha_n$ in that order within the interval $[0,t)$ is given by the limit
\begin{align}
&\lim_{m\to\infty}\int_{C_W^m}\sum_{1\leq j_1<\cdots<j_n\leq m}\prod_{\substack{l=1\\l\notin\{j_i\}_{i=1}^n}}^m\Big[1-\frac{\beta}{2}\sum_{\bxi\in R_+}\kappa(\bxi)\frac{\|\bxi\|^2}{2}\frac{\Delta t}{(\bxi\cdot\by_l)^2}\Big]\nonumber\\
&\quad\quad\times\prod_{l=1}^n\frac{\beta}{2}\kappa(\balpha_l)\frac{\|\balpha_l\|^2}{2}\frac{\Delta t}{(\balpha_l\cdot\tau_{l-1}\by_{j_l})^2}\times\prod_{l=1}^m\hat{p}(t_l,\by_{l}|\bx_0)\intd^N \by_l\nonumber\\
&=\mnm{e}^{-\int_0^t{\Lambda_\beta}(s|\bx_0)\intd s}\int_{0\leq s_1<\cdots<s_n\leq t}\prod_{j=1}^n{\lambda_\beta}(s_j,\tau_{j-1}^{-1}\balpha_j|\bx_0)\intd s_j.
\end{align}
Then, the probability $P^{\mcl{J}}_{\beta}(t,\tau|\bx_0)$ is given by the sum over all possible sequences of reflections such that $\tau=\tau_n=\sigma_{n}\cdots\sigma_1$ for $n$ a non-negative integer, namely,
\begin{align}
P^{\mcl{J}}_{\beta}(t,\tau|\bx_0)=&\mnm{e}^{-\int_0^t{\Lambda_\beta}(s|\bx_0)\intd s}\nonumber\\
&\times\sum_{n=0}^\infty\sum_{\substack{\{\balpha_i\in R_+\}_{i=1}^n:\\\tau_n=\tau}}\int_{0\leq s_1<\cdots<s_n\leq t}\prod_{j=1}^n{\lambda_\beta}(s_j,\tau_{j-1}^{-1}\balpha_j|\bx_0)\intd s_j.
\end{align}
Differentiating with respect to time yields
\begin{align}
\frac{\partial}{\partial t}P^{\mcl{J}}_{\beta}(t,\tau|\bx_0)=&\mnm{e}^{-\int_0^t{\Lambda_\beta}(s|\bx_0)\intd s}\sum_{n=1}^\infty\sum_{\balpha\in R_+}{\lambda_\beta}(t,\tau_{n-1}^{-1}\balpha|\bx_0)\nonumber\\
&\times\sum_{\substack{\{\balpha_i\in R_+\}_{i=1}^{n-1}:\\\tau_{n-1}=\sigma_{\balpha}\tau}}\int_{0\leq s_1<\cdots<s_{n-1}\leq t}\prod_{j=1}^{n-1}{\lambda_\beta}(s_j,\tau_{j-1}^{-1}\balpha_j|\bx_0)\intd s_j\nonumber\\
&-{\Lambda_\beta}(t|\bx_0)P^{\mcl{J}}_{\beta}(t,\tau|\bx_0)\nonumber\\
=&\mnm{e}^{-\int_0^t{\Lambda_\beta}(s|\bx_0)\intd s}\sum_{\balpha\in R_+}{\lambda_\beta}(t,\tau^{-1}\sigma_{\balpha}\balpha|\bx_0)\nonumber\\
&\times\sum_{n=1}^\infty\sum_{\substack{\{\balpha_i\in R_+\}_{i=1}^{n-1}:\\\tau_{n-1}=\sigma_{\balpha}\tau}}\int_{0\leq s_1<\cdots<s_{n-1}\leq t}\prod_{j=1}^{n-1}{\lambda_\beta}(s_j,\tau_{j-1}^{-1}\balpha_j|\bx_0)\intd s_j\nonumber\\
&-{\Lambda_\beta}(t|\bx_0)P^{\mcl{J}}_{\beta}(t,\tau|\bx_0)\nonumber\\
=&\sum_{\balpha\in R_+}{\lambda_\beta}(t,\tau^{-1}\balpha|\bx_0)P^{\mcl{J}}_{\beta}(t,\sigma_{\balpha}\tau|\bx_0)-{\Lambda_\beta}(t|\bx_0)P^{\mcl{J}}_{\beta}(t,\tau|\bx_0).
\end{align}
Finally, we can simplify the first term by recalling that $\sigma_{\balpha}\tau=\tau \sigma_{\tau^{-1}\balpha}$, yielding
\begin{align}
\sum_{\balpha\in R_+}{\lambda_\beta}(t,\tau^{-1}\balpha|\bx_0)P^{\mcl{J}}_{\beta}(t,\sigma_{\balpha}\tau|\bx_0)&=\frac{1}{2}\sum_{\balpha\in R}{\lambda_\beta}(t,\tau^{-1}\balpha|\bx_0)P^{\mcl{J}}_{\beta}(t,\tau \sigma_{\tau^{-1}\balpha}|\bx_0)\notag\\
&=\sum_{\balpha\in R_+}{\lambda_\beta}(t,\balpha|\bx_0)P^{\mcl{J}}_{\beta}(t,\tau \sigma_{\balpha}|\bx_0).
\end{align}
The second equality follows from the substitution $\balpha^\prime=\tau^{-1}\balpha$ and from the invariance of $R$ relative to the action of $W$. Inserting this expression in the master equation gives the desired result.
\end{proof}

It is clear that the master equation preserves the total sum of the probabilities, and that a possible equilibrium state is that in which $P^{\mcl{J}}_{\beta}(t,\tau|\bx_0)=1/|W|$, as it makes the rhs vanish. Moreover, it is easy to see that \eqref{eq:master} is stable, in the sense that its rhs is given by a linear and negative semidefinite operator mapping the space of real functions on $W$ to itself. To see this, we consider arbitrary functions $f,g:W\to\mbb{R}$ and define the inner product 
\begin{equation}
(f,g):=\sum_{\tau\in W}f(\tau)g(\tau).
\end{equation}
Then, we see that, if we define the operator $\mcl{M}$ so that we can write the rhs of the master equation in the form $\mcl{M}P^{\mcl{J}}_{\beta}$, we get
\begin{align}\label{eq:negsemidef}
(f,\mcl{M}f)&=-\sum_{\tau\in W}f(\tau)\Big[{\Lambda_\beta}(t|\bx_0)f(\tau)-\sum_{\balpha\in R_+}{\lambda_\beta}(t,\balpha|\bx_0)f(\tau\sigma_{\balpha})\Big]\nonumber\\
&=-\sum_{\tau\in W}f(\tau)\sum_{\balpha\in R_+}{\lambda_\beta}(t,\balpha|\bx_0)[f(\tau)-f(\tau\sigma_{\balpha})]\nonumber\\
&=-\frac{1}{2}\sum_{\tau\in W}\sum_{\balpha\in R_+}{\lambda_\beta}(t,\balpha|\bx_0)[f(\tau)-f(\tau\sigma_{\balpha})]^2.
\end{align}
Because ${\lambda_\beta}(t,\balpha|\bx_0)>0$, the expression is negative semidefinite, and it is only equal to zero when $f(\tau)=f(\tau\sigma_{\balpha})$ for every root $\balpha$, which only occurs if $f(\tau)$ is a constant function. Therefore, the equilibrium state is unique. We now turn to the jump counting process to investigate the dynamics of $\rho(t)$ in more detail.

\section{The Dunkl jump counting process\\and the jump rate phase transition}\label{sec:phasetransition}

Denote by $\mcl{N}(t)$ the number of discontinuities in the path of $\bX$, namely, the total number of jumps that $\rho(t)$ has carried out up to time $t$. Since the probability rate of a jump is given by \eqref{eq:totalrate}, whenever $\beta k(\balpha)>1$ for every $\balpha \in R_+$ the jump rate process is an inhomogeneous Poisson process with rate function $\Lambda_\beta(t|\bx_0)$. The following is straightforward.
\begin{lemma}
Denote by $P^{\mcl{N}}_{\beta}(t,n|\bx_0)$ the probability that $\mcl{N}(t)=n$, given that the associated Dunkl process started at $\bx_0$. Then, 
\begin{equation}
\frac{\partial}{\partial t}P^{\mcl{N}}_{\beta}(t,n|\bx_0)={\Lambda_\beta}(t|\bx_0)[P^{\mcl{N}}_{\beta}(t,n-1|\bx_0)-P^{\mcl{N}}_{\beta}(t,n|\bx_0)].
\end{equation}
\end{lemma}
\begin{proof}
The proof is similar to that of Lemma~\ref{lm:masterequation}. By discretizing the interval $[0,t)$ as before, one can show that
\begin{equation}
P^{\mcl{N}}_{\beta}(t,n|\bx_0)=\frac{1}{n!}\Big[\int_0^t{\Lambda_\beta}(s|\bx_0)\intd s\Big]^n\mnm{e}^{-\int_0^t{\Lambda_\beta}(s|\bx_0)\intd s},
\end{equation}
and taking a time derivative yields the result.
\end{proof}
This result shows that the jump counting process depends solely on the properties of the total jump rate, ${\Lambda_\beta}(t|\bx_0)$. We will focus on this quantity for the rest of the section.

We note now that all of the jump rates $\lambda_\beta(t,\balpha|\bx_0)$ and, by extension, ${\Lambda_\beta}(t|\bx_0)$ satisfy a crucial scaling property which stems from \eqref{eq:TPDScaling}:
\begin{equation}\label{eq:ratescaling}
\lambda_\beta(t,\balpha|\bx_0)=\frac{\beta\|\balpha\|^2}{4}\int_{C_W}\frac{k(\balpha)}{(\balpha\cdot\bx)^2}\hat{p}\Big(1,\frac{\bx}{\sqrt{t}}\Big|\frac{\bx_0}{\sqrt{t}}\Big)\frac{\ud^N \bx}{t^{N/2}}=\frac{1}{t}\lambda_\beta\Big(1,\balpha\Big|\frac{\bx_0}{\sqrt{t}}\Big),
\end{equation}
whenever $t>0$. In fact, we know from the remarks following Theorem~1 in \cite{AndrausMiyashita15} that in general
\begin{equation}\label{eq:jumpratefirstorder}
\lambda_\beta\Big(1,\balpha\Big|\frac{\bx_0}{\sqrt{t}}\Big)=\frac{\beta\|\balpha\|^2|W|}{4c_\beta}\int_{C_W}\frac{k(\balpha)}{(\balpha\cdot\bx)^2}\mnm{e}^{-\|\bx\|^2/2}w_\beta(\bx)\ud^N \bx+O(\|\bx_0\|^2/t),
\end{equation}
so for ${\Lambda_\beta}(t|\bx_0)$, we have
\begin{equation}\label{eq:totalratescaling}
{\Lambda_\beta}(t|\bx_0)=\frac{1}{t}{\Lambda_\beta}(1|\mb{0})+O(\|\bx_0\|^2/t^2)
\end{equation}
assuming that $\bx_0$ belongs to the space spanned by $R_+$. From this relation it is clear that the decrease of the total jump rate in time is extremely slow, at least one order in $t$ slower than the effect of $\bx_0$. Consequently, the expected value of $\mcl{N}(t)$, which corresponds to the time integral of ${\Lambda_\beta}(t|\bx_0)$ itself, shows a logarithmic behavior at long times. This means that there may be meaningful dynamics in the jump process long after the transient effect of the initial configuration disappears, and that this long time behavior could be universal among all types of Dunkl jump processes. It is of interest, then, to calculate ${\Lambda_\beta}(1|\mb{0})$ explicitly; we can obtain an expression for it when $k(\balpha)\equiv1$.
\begin{lemma}\label{lm:totaljumprate}
Suppose that $\beta>1$ and that $k(\balpha)=1$ for all $\balpha\in R_+$. Then, 
\begin{equation}
{\Lambda_\beta}(1|\mb{0})=\frac{\beta |R_+|}{4(\beta-1)}.
\end{equation}
\end{lemma}

\begin{proof}
By definition,
\begin{equation}
{\Lambda_\beta}(1|\mb{0})=\int_{\mbb{R}^N}\sum_{\balpha\in R_+}\frac{\beta\|\balpha\|^2}{4c_\beta(\balpha\cdot\bx)^2}\text{e}^{-\|\bx\|^2/2}w_\beta(\bx)\intd^N \bx.
\end{equation}
Note that, because $k(\balpha)\equiv1$,
\begin{align}
\sum_{i=1}^N\frac{\partial^2}{\partial x_i^2}w_\beta(\bx)&=\sum_{i=1}^N\frac{\partial}{\partial x_i}\Big(w_\beta(\bx)\sum_{\balpha\in R_+}\frac{\beta \alpha_i}{\balpha\cdot\bx}\Big)\notag\\
&=w_\beta(\bx)\beta(\beta-1)\sum_{\balpha\in R_+}\frac{\|\balpha\|^2}{(\balpha\cdot\bx)^2}.
\end{align}
The second line follows from Lemma~4.4.6 in \cite{DunklXu}. Then, we can write
\begin{equation}
{\Lambda_\beta}(1|\mb{0})=\frac{1}{4c_\beta(\beta-1)}\int_{\mbb{R}^N}\text{e}^{-\|\bx\|^2/2}\sum_{i=1}^N\frac{\partial^2}{\partial x_i^2}w_\beta(\bx)\intd^N \bx.
\end{equation}
Recalling that this integral converges for $\beta>1$, we integrate by parts twice to obtain
\begin{align}
{\Lambda_\beta}(1|\mb{0})&=\frac{1}{4c_\beta(\beta-1)}\int_{\mbb{R}^N}(\|\bx\|^2-N)\text{e}^{-\|\bx\|^2/2}w_\beta(\bx)\intd^N \bx\notag\\
&=\frac{1}{4c_\beta(\beta-1)}\int_{\mbb{R}^N}\|\bx\|^2\text{e}^{-\|\bx\|^2/2}w_\beta(\bx)\intd^N \bx-\frac{N}{4(\beta-1)}.
\end{align}
For the last integral, consider the function
\begin{equation}
f(a):=\int_{\mbb{R}^N}\text{e}^{-a\|\bx\|^2/2}w_\beta(\bx)\intd^N \bx,\quad a>0.
\end{equation}
One can readily calculate this integral by setting $\by=\sqrt{a}\bx$ to get
\begin{equation}
f(a)=a^{-(N+\beta |R_+|)/2}\int_{\mbb{R}^N}\text{e}^{-\|\by\|^2/2}w_\beta(\by)\intd^N \by=c_\beta a^{-(N+\beta |R_+|)/2}.
\end{equation}
Then, we see that
\begin{equation}
\int_{\mbb{R}^N}\|\bx\|^2\text{e}^{-\|\bx\|^2/2}w_\beta(\bx)\intd^N \bx=-2\frac{\ud}{\ud a}f(a)\Big|_{a=1}=c_\beta(N+\beta |R|_+).
\end{equation}
Inserting this result in our expression for ${\Lambda_\beta}(1|\mb{0})$ yields the statement.
\end{proof}

It is of great interest to investigate this result from the perspective of infinite particle systems. The radial Dunkl processes of types $A_{N-1}$ and $B_N$ correspond to the Dyson model \cite{Dyson62} and the Wishart-Laguerre processes \cite{Bru,KonigOConnell} respectively. In the particular case $\beta=2$, some of their properties in the infinite-particle limit have been elucidated in \cite{KatoriTanemura07,KatoriTanemura11,OsadaTanemura14}, in particular in the \emph{bulk scaling limit}, where the process time $t$ is scaled linearly with the number of particles, $N$. Under these conditions, we give the proof of Theorem~\ref{th:phasetransition}.

\begin{proof}[Proof of  Theorem~\ref{th:phasetransition}]
We start by using \eqref{eq:totalratescaling} to write
\begin{equation}
\frac{1}{N}\Lambda_\beta^{(R)}(N|\bx_0^{(N)})=\frac{1}{N^2}{\Lambda_\beta}^{(R)}(1|\mb{0})+O(|R_+|/N^3).
\end{equation}
The factor of $|R_+|$ in the correction term appears because the total jump rate is the sum of $|R_+|$ finite jump rates, and we have used the hypothesis $\|\bx_0\|<K$. Then, by Lemma~\ref{lm:totaljumprate} we obtain
\begin{equation}
\frac{1}{N}\Lambda_\beta^{(R)}(N|\bx_0^{(N)})=\frac{1}{N^2}\frac{\beta |R_+|}{4(\beta-1)}+O(|R_+|/N^3).
\end{equation}
For the case $R=A_{N-1}$, we have $|A_{N-1,+}|=N(N-1)/2$, so
\begin{equation}
\frac{1}{N}\Lambda_\beta^{(A)}(N|\bx_0^{(N)})=\frac{\beta N(N-1)}{8N^2(\beta-1)}+O(N^{-1})\stackrel{N\to\infty}{\longrightarrow}\frac{\beta}{8(\beta-1)}.
\end{equation}
Similarly, $|B_{N,+}|=N^2$, and
\begin{equation}
\frac{1}{N}\Lambda_\beta^{(B)}(N|\bx_0^{(N)})=\frac{\beta N^2}{4N^2(\beta-1)}+O(N^{-1})\stackrel{N\to\infty}{\longrightarrow}\frac{\beta}{4(\beta-1)},
\end{equation}
as desired.
\end{proof}

This phase transition stands in contrast to that reported in \cite{ValkoVirag09} for the bulk scaling limit of the Dyson model. There, the phase transition corresponds to the asymptotic behavior of the stochastic sine equation, while the phase transition presented here is related to the jump rates, and in turn, to the collision probability in either the Dyson model or the Wishart-Laguerre processes. Because the first collision time with the walls of $C_W$ is almost-surely finite whenever $\beta k(\balpha)<1$ for any $\balpha\in R_+$ \cite{Demni08}, if the corresponding radial Dunkl process hits a wall at least one of the jump rates diverges, which explains the singularity at $\beta_c$ of the total jump rate per particle in the bulk scaling limit.

\section{Relaxation behavior}\label{sec:relaxation}

Due to the long-time behavior of the total jump rate, it is expected that the relaxation to equilibrium be a very slow process. This observation, paired with \eqref{eq:ratescaling} means that the relaxation behavior ultimately depends on the large-$t$ solution of \eqref{eq:master}, in which the effect of the initial configuration plays an important role. We clarify the situation by turning to the proof of Theorem~\ref{th:relaxation}.

\begin{proof}[Proof of Theorem~\ref{th:relaxation}]
We start by modifying Lemma~\ref{lm:masterequation} using \eqref{eq:ratescaling} to derive the relaxation asymptotics. By \eqref{eq:jumpratefirstorder}, we see that
\begin{equation}
\lambda_\beta(t,\balpha|\bx_0)=\frac{1}{t}\lambda_\beta(1,\balpha|\mb{0})+O(\|\bx_0\|^2/t^{2}),
\end{equation}
so we choose a time scale $t_0$ such that $\epsilon^2>\|\bx_0\|^2/t_0$ for $\epsilon>0$. Then, for $t>t_0$, the master equation reads
\begin{align}
\frac{\partial}{\partial t}P^{\mcl{J}}_{\beta}(t,\tau|\bx_0)=&\frac{1}{t}\Big[\sum_{\balpha\in R_+}\lambda_\beta(1,\balpha|\mb{0})P^{\mcl{J}}_{\beta}(t,\tau\sigma_{\balpha}|\bx_0)-\Lambda_\beta(1|\mb{0})P^{\mcl{J}}_{\beta}(t,\tau|\bx_0)\Big]\notag\\
&+O(|R_+|\epsilon^2 t_0/t^{2}).
\end{align}
Let us consider now the initial value problem
\begin{equation}
\frac{\partial}{\partial t}P^{\mcl{L}}_{\beta}(t,\tau|\bx_0)=\frac{1}{t}\Big[\sum_{\balpha\in R_+}\lambda_\beta(1,\balpha|\mb{0})P^{\mcl{L}}_{\beta}(t,\tau\sigma_{\balpha}|\bx_0)-\Lambda_\beta(1|\mb{0})P^{\mcl{L}}_{\beta}(t,\tau|\bx_0)\Big],
\end{equation}
for $t\geq t_0$ and $P^{\mcl{L}}_{\beta}(t_0,\tau|\bx_0)=P^{\mcl{J}}_{\beta}(t_0,\tau|\bx_0)$ for every $\tau\in W$.  Let us define the matrix
\begin{equation}
M_{\tau\nu}:=\sum_{\balpha\in R_+}\lambda_\beta(1,\balpha|\mb{0})\delta_{\tau\sigma_{\balpha},\nu}-\Lambda_\beta(1|\mb{0})\delta_{\tau,\nu},
\end{equation}
where $\delta_{\tau,\nu}=1$ when $\tau=\nu$ and zero otherwise for $\tau,\nu\in W$ (the Kronecker delta function). Then, we can rewrite the master equation above as
\begin{equation}
\frac{\partial}{\partial t}P^{\mcl{L}}_{\beta}(t,\tau|\bx_0)=\frac{1}{t}\sum_{\nu\in W}M_{\tau\nu}P^{\mcl{L}}_{\beta}(t,\nu|\bx_0),
\end{equation}
and we can see that $M_{\tau\nu}$ is negative semidefinite and symmetric, as this is simply the matrix form of a particular case of the operator $\mcl{M}$ defined previously. Then, $M_{\tau\nu}$ is diagonalizable and all of its eigenvalues are non-positive, so we decompose $P^{\mcl{L}}_{\beta}(t,\nu|\bx_0)$ in terms of the orthonormal eigenvector basis $\{\phi_i\}_{i=1}^{|W|}$ of $M_{\tau\nu}$ with eigenvalue $r_i\leq 0$:
\begin{equation}
P^{\mcl{L}}_{\beta}(t,\tau|\bx_0)=\sum_{i=1}^{|W|}K_{i}(\bx_0) g_{r_i}(t)\phi_{i}(\tau).
\end{equation}
With this, we obtain
\begin{equation}
\sum_{i=1}^{|W|}K_{i}(\bx_0) \phi_{i}(\tau)\frac{\ud}{\ud t}g_{r_i}(t)=\sum_{i=1}^{|W|}K_{i}(\bx_0) g_{r_i}(t) \frac{r_i}{t}\phi_{i}(\tau),
\end{equation}
and by the orthogonality of the eigenvectors, we obtain the following differential equation for the time-dependent part,
\begin{equation}
\frac{\ud}{\ud t}g_{r_i}(t)=\frac{r_i}{t}g_{r_i}(t),
\end{equation}
which is readily integrated to obtain
\begin{equation}
g_{r_i}(t)=\Big(\frac{t}{t_0}\Big)^{r_i}.
\end{equation}
Here, we have chosen the integration constant to be one, as the initial condition is imposed on the constants $K_{i}(\bx_0)$,
\begin{equation}
P^{\mcl{L}}_{\beta}(t_0,\tau|\bx_0)=\sum_{j=1}^{|W|}K_{j}(\bx_0) \phi_{j}(\tau), 
\end{equation}
so we write
\begin{equation}\label{eq:relaxationcoefficients}
\sum_{\tau\in W}\phi_i(\tau)P^{\mcl{J}}_{\beta}(t_0,\tau|\bx_0)=\sum_{j=1}^{|W|}K_{j}(\bx_0) \sum_{\tau\in W}\phi_i(\tau)\phi_j(\tau)=K_{i}(\bx_0).
\end{equation}
In the same way as in \eqref{eq:negsemidef}, one can show that there is a unique eigenvector of $M_{\tau\nu}$ for the eigenvalue $0$, given by $1/\sqrt{|W|},$ so
\begin{equation}
\sum_{\tau\in W}\frac{1}{\sqrt{|W|}}P^{\mcl{J}}_{\beta}(t_0,\tau|\bx_0)=\frac{1}{\sqrt{|W|}}
\end{equation}
because $P^{\mcl{J}}_{\beta}(t_0,\tau|\bx_0)$ is normalized, and finally, we obtain
\begin{equation}\label{eq:longtimesolution}
P^{\mcl{L}}_{\beta}(t,\tau|\bx_0)=\frac{1}{|W|}+\sum_{i:r_i<0}K_i(\bx_0) \Big(\frac{t}{t_0}\Big)^{r_i} \phi_i(\tau).
\end{equation}
Now, from the time derivative of the difference $P^{\mcl{J}}_{\beta}(t,\tau|\bx_0)-P^{\mcl{L}}_{\beta}(t,\tau|\bx_0)$ for $t>t_0$, we obtain
\begin{equation}
P^{\mcl{J}}_{\beta}(t,\tau|\bx_0)-P^{\mcl{L}}_{\beta}(t,\tau|\bx_0)=O(|R_+|\epsilon^2 t_0/t).
\end{equation}
We see that the correction is of order $\epsilon^2/t$. This means that $P^{\mcl{J}}_{\beta}(t,\tau|\bx_0)$ follows closely the behavior of $P^{\mcl{L}}_{\beta}(t,\tau|\bx_0)$, and therefore relaxes to the equilibrium configuration by the power law given by \eqref{eq:longtimesolution} if the least negative eigenvalue $r$ satisfies $|r|<1$, or by a power law with exponent $-1$ otherwise.
\end{proof}

In spite of this theorem, the speed of relaxation remains unclear without knowing the eigenvalues of $M_{\tau\nu}$. We can make a few immediate observations which give some insight into the structure of the master equation at long times. Let us define the matrix $\tilde{M}_{\tau\nu}$ by
\begin{equation}
\tilde{M}_{\tau\nu}:=\sum_{\balpha\in R_+}\lambda_\beta(1,\balpha|\mb{0})\delta_{\tau\sigma_{\balpha},\nu}.
\end{equation}
\begin{lemma}\label{lm:symmetricspectrum}
The matrix $\tilde{M}_{\tau\nu}$ has a symmetric spectrum, namely, if $r$ is one of its eigenvalues, then $-r$ is an eigenvalue as well.
\end{lemma}
\begin{proof}
First, let us suppose that an eigenvector with eigenvalue $r$ is denoted by $\phi_r(\tau)$. Denoting the signature of $\tau\in W$ by $\sgn(\tau)$ (namely, the determinant of its matrix representation), let us consider the vector $\sgn (\tau)\phi_r(\tau)$:
\begin{align}
\sum_{\nu\in W}\tilde{M}_{\tau\nu}\ \sgn (\nu)\phi_r(\nu)&=\sum_{\balpha\in R_+}\lambda_\beta(1,\balpha|\mb{0})\ \sgn(\tau\sigma_{\balpha})\phi_r(\tau\sigma_{\balpha})\notag\\
&=-\sgn(\tau)\sum_{\balpha\in R_+}\lambda_\beta(1,\balpha|\mb{0})\ \phi_r(\tau\sigma_{\balpha})\notag\\
&=-r\ \sgn(\tau)\phi_r(\tau).
\end{align}
It follows that $-r$ is an eigenvalue with eigenvector $\sgn(\tau)\phi_r(\tau)$.
\end{proof}

From this fact, we obtain the most negative eigenvalue of the matrix $M_{\tau\nu}$ as follows.
\begin{lemma}\label{lm:minimumeigvl}
The minimum (most negative) relaxation exponent of a Dunkl jump process with equal multiplicities is given by
\begin{equation}
r_\textnormal{min}=-\frac{\beta|R_+|}{2(\beta-1)}.
\end{equation}
\end{lemma}
\begin{proof}
The relaxation exponents are the eigenvalues of the matrix $M_{\tau\nu}$, which can be rewritten as
\begin{equation}
M_{\tau\nu}=\tilde{M}_{\tau\nu}-\Lambda_\beta(1|\mb{0})\delta_{\tau\nu}.
\end{equation}
Since we know that $M_{\tau\nu}$ is negative semidefinite, and that its largest eigenvalue is zero with eigenvector $\phi_0(\tau)=1/\sqrt{|W|}$, it follows that the largest eigenvalue of the matrix $\tilde{M}_{\tau\nu}$ is exactly $\Lambda_\beta(1|\mb{0})$. By Lemma~\ref{lm:symmetricspectrum}, the vector $\phi_{\textnormal{min}}(\tau)=\sgn(\tau)/\sqrt{|W|}$ gives the minimum eigenvalue of $\tilde{M}_{\tau\nu}$, namely, $-\Lambda_\beta(1|\mb{0})$, so it is clear that
\begin{equation}
\sum_{\nu\in W}M_{\tau\nu}\phi_{\textnormal{min}}(\nu)=-2\Lambda_{\beta}(1|\mb{0})\phi_\textnormal{min}(\tau).
\end{equation}
Using Lemma~\ref{lm:totaljumprate} completes the proof.
\end{proof}
For the rest of the eigenvalues, we need to calculate the individual jump rates given by \eqref{eq:jumprates}; these are difficult to calculate, even in the case $\lambda_\beta(1,\balpha|\mb{0})$, due to the fact that the integrand is not $W$-invariant. However, the calculations carried out up to this point reveal a fact that is meaningful from a physical viewpoint. The equilibrium state of the Dunkl jump process $\phi_0(\tau)=1/\sqrt{|W|}$ indicates a completely symmetric state, and the state with the fastest relaxation exponent, $\phi_\textnormal{min}(\tau)=\sgn(\tau)/\sqrt{|W|}$ indicates a completely \emph{antisymmetric} state. In this sense, the universal behavior of the process consists of moving away from a fermionic configuration and converging to a bosonic configuration.

These observations are in contrast with the behavior of the paths of the corresponding radial Dunkl process: if we take $\beta=2$ we obtain well-known determinantal processes for the root systems of types $A_{N-1}$ and $B_N$ \cite{KatoriTanemura07, KatoriTanemura11}. Therefore, the behavior of these radial Dunkl processes can be viewed as fermionic in nature, and this behavior is underlined by the fact that when $\beta=2$ these processes admit a formulation as free Brownian motions (for the case $A_{N-1}$) or free Bessel processes (for the case $B_N$) conditioned never to collide in the sense of Doob. Furthermore, this is a behavior that persists throughout the process, while the Dunkl jump processes considered here present a dynamical behavior which goes from fermionic to bosonic for all $\beta>1$.

We can extract more information from the asymptotics for large $\beta$. Lemma~\ref{lm:totaljumprate} in the case where $\beta\to\infty$ gives a total jump rate which is positive,
\begin{equation}\label{eq:frozentotalrate}
\lim_{\beta\to\infty}\frac{1}{t}{\Lambda_\beta}(1|\mb{0})=\frac{1}{t}\lim_{\beta\to\infty}\frac{\beta |R_+|}{4(\beta-1)}=\frac{1}{t}\frac{|R_+|}{4},
\end{equation}
and this value is the maximum lower bound of the total jump rate for finite $\beta>1$ at long times. It is known \cite{AndrausMiyashita15, AndrausKatoriMiyashita12, AndrausKatoriMiyashita14, AndrausVoit18, VoitWoerner, AndrausVoit19} that for $\beta\to\infty$, the continuous part of a Dunkl process follows a deterministic path with small Gaussian fluctuations. However, for the Dunkl jump process not only does the jump rate stay bounded, it converges to its minimum value, meaning that the relaxation of the jump process is slowest in this regime. Consequently, the jump rates ${\lambda_\beta}(1,\balpha|\mb{0})$ must also converge to well-defined values. As seen in \cite{AndrausMiyashita15}, the probability density $\hat{p}(1,\bx|\mb{0})$ concentrates at the peak set of the root system in question after scaling by a factor of $\sqrt{\beta}$ \cite{Dunkl89B}. The peak set is the set of vectors which minimize the function
\begin{equation}
\frac{\|\bx\|^2}{2}-\sum_{\balpha\in R_+}k(\balpha)\log|\balpha\cdot\bx|.
\end{equation}
Due to the $W$-invariance of this function, there are $|W|$ such minimizers, so we choose the minimizer which lies in $C_W$ and denote it by $\bz^{(R)}$. By the minimizing condition, $\bz^{(R)}$ satisfies
\begin{equation}\label{eq:peakvector}
\bz^{(R)}=\sum_{\balpha\in R_+}\frac{k(\balpha)\balpha}{\balpha\cdot\bz^{(R)}}.
\end{equation}
Then, we have
\begin{equation}\label{eq:frozenrates}
\lim_{\beta\to\infty}\lambda_\beta(1,\balpha|\mb{0})=\frac{\|\balpha\|^2 k(\balpha)}{4}\lim_{\beta\to\infty}\int_{C_W}\frac{\hat{p}(1,\sqrt{\beta}\bx|\mb{0})}{(\balpha\cdot\bx)^2}\beta^{N/2}\intd^N \bx=\frac{\|\balpha\|^2 k(\balpha)}{4(\balpha\cdot\bz^{(R)})^2},
\end{equation}
and
\begin{equation}
\lim_{\beta\to\infty}\Lambda_\beta(1|\mb{0})=\sum_{\balpha\in R_+}\frac{\|\balpha\|^2 k(\balpha)}{4(\balpha\cdot\bz^{(R)})^2}.
\end{equation}
It is straightforward to check that this expression is consistent with \eqref{eq:frozentotalrate} when $k(\balpha)\equiv1$. We can now prove the following statement on the jump rate asymptotics.

\begin{lemma}\label{lm:coefficientasymptotics}
Define $r^\star:=\min_{\bzeta\in R_+}\bzeta\cdot\bz^{(R)}/\|\bzeta\|\sqrt{\gamma}$, with $\gamma$ given in \eqref{eq:gammadefinition}. There exist constants $M>0$ and $\tilde{C}(\balpha)>0$ such that the jump rates are given by
\begin{equation}
\lambda_\beta(1,\balpha|\mb{0})=\frac{\|\balpha\|^2 k(\balpha)}{4(\balpha\cdot\bz^{(R)})^2}\Big[1+\frac{\tilde{C}(\balpha)}{\beta}+O(\beta^{-2})+O(\mnm{e}^{-M})\Big]
\end{equation}
for $\beta>\min(2M/\gamma r^\star,1)$.
\end{lemma}
\begin{proof}
The proof consists of a large-$\beta$ expansion. For every $\by\in C_W$, let us define $C_W(\by):=\{\bx\in\mbb{R}^N:\bx+\by\in C_W\}$. We also set $\bz:=\bz^{(R)}$ and $\lambda_{\balpha}(\beta):=\lambda_\beta(1,\balpha|\mb{0})$ for simplicity. By \eqref{eq:jumprates}, we have
\[\lambda_{\balpha}(\beta)=\frac{\beta\|\balpha\|^2}{4}\frac{|W|k(\balpha)}{c_\beta}\int_{C_W(\sqrt{\beta}\bz)}\frac{\text{e}^{-\|\bx+\sqrt{\beta}\bz\|^2/2}}{[\balpha\cdot(\bx+\sqrt{\beta}\bz)]^2}w_\beta(\bx+\sqrt{\beta}\bz)\intd^N \bx\]
after shifting the Weyl chamber by $\sqrt{\beta}\bz$. Note that
\begin{align}
[\balpha\cdot(\bx+\sqrt{\beta}\bz)]^2&=\beta(\balpha\cdot\bz)^2\Big(1+\frac{\balpha\cdot\bx}{\sqrt{\beta}\balpha\cdot\bz}\Big)^2,\notag\\
\|\bx+\sqrt{\beta}\bz\|^2&=\|\bx\|^2+\beta\gamma+2\sqrt{\beta}\bx\cdot\bz,\text{ and}\notag
\end{align}
\begin{align}
\sum_{\bzeta\in R_+}k(\bzeta)\log[\bzeta\cdot(\bx+\sqrt{\beta}\bz)]&=\frac{\gamma}{2}\log \beta+\sum_{\bzeta\in R_+}k(\bzeta)\log\bzeta\cdot\bz\notag\\
&\quad+\sum_{\bzeta\in R_+}k(\bzeta)\log\Big(1+\frac{\bzeta\cdot\bx}{\sqrt{\beta}\bzeta\cdot\bz}\Big),\notag
\end{align}
where we note that $\|\bz\|^2=\gamma$. Then, we can write
\begin{align}
\lambda_{\balpha}(\beta)=&\frac{\|\balpha\|^2k(\balpha)}{4(\balpha\cdot\bz)^2}\Big(\frac{\beta}{\text{e}}\Big)^{\beta\gamma/2}w_\beta(\bz)\frac{|W|}{c_\beta}\notag\\
&\times \int_{C_W(\sqrt{\beta}\bz)}\frac{\text{e}^{-\|\bx\|^2/2-\sqrt{\beta}\bx\cdot\bz}}{[1+\frac{\balpha\cdot\bx}{\sqrt{\beta}\balpha\cdot\bz}]^2}\prod_{\bzeta\in R_+}\Big(1+\frac{\bzeta\cdot\bx}{\sqrt{\beta}\bzeta\cdot\bz}\Big)^{\beta k(\bzeta)}\intd^N \bx.\label{eq:mainintegral}
\end{align}
We can divide the integral in two regions: $\bx\in \mcl{I}(\beta):=\{\by\in\mbb{R}^N:|\bxi\cdot\by|<\sqrt{\beta}\bxi\cdot\bz,\ {}^\forall\bxi\in R_+\}$ and $\bx\in C_W(\sqrt{\beta}\bz)\setminus\mcl{I}(\beta)$.
With this, we can use the expansion formula for $\log(1+x)$
\begin{align}\label{eq:logexpansionw}
\sum_{\bzeta\in R_+}k(\bzeta)\log\Big(1+\frac{\bzeta\cdot\bx}{\sqrt{\beta}\bzeta\cdot\bz}\Big)=\sum_{\bzeta\in R_+}k(\bzeta)\sum_{n=1}^\infty \frac{(-1)^{n-1}}{n\beta^{n/2}}\Big(\frac{\bzeta\cdot\bx}{\bzeta\cdot\bz}\Big)^n,
\end{align}
and the expansion
\begin{align}
\Big(1+\frac{\balpha\cdot\bx}{\sqrt{\beta}\balpha\cdot\bz}\Big)^{-2}=\sum_{m=0}^\infty\frac{(-1)^m (m+1)}{\beta^{m/2}} \Big(\frac{\balpha\cdot\bx}{\balpha\cdot\bz}\Big)^m
\end{align}
within $\mcl{I}(\beta)$. This is a convex set, and it is symmetric in the sense that if $\bx\in\mcl{I}(\beta)$, then $-\bx\in\mcl{I}(\beta)$. The integral in this region becomes
\begin{align}
I_1:=&\int_{\mcl{I}(\beta)}\frac{\text{e}^{-\|\bx\|^2/2-\sqrt{\beta}\bx\cdot\bz}}{[1+\frac{\balpha\cdot\bx}{\sqrt{\beta}\balpha\cdot\bz}]^2}\prod_{\bzeta\in R_+}\Big(1+\frac{\bzeta\cdot\bx}{\sqrt{\beta}\bzeta\cdot\bz}\Big)^{\beta k(\bzeta)}\intd^N \bx\notag\\
=&\int_{\mcl{I}(\beta)}\text{e}^{-\bx^T\mb{H}\bx/2}\sum_{m=0}^\infty\frac{(-1)^m (m+1)}{\beta^{m/2}} \Big(\frac{\balpha\cdot\bx}{\balpha\cdot\bz}\Big)^m\notag\\
&\times\exp\Big[\sum_{\bzeta\in R_+}k(\bzeta)\sum_{l=3}^\infty\frac{(-1)^{l-1}}{l\beta^{(l-2)/2}}\Big(\frac{\bzeta\cdot\bx}{\bzeta\cdot\bz}\Big)^l\Big]\intd^N \bx.
\end{align}
For the second line, we have defined the Hessian matrix by
\begin{equation}
\mb{H}=\sum_{\bzeta\in R_+}k(\bzeta)\frac{\bzeta\bzeta^T}{(\bzeta\cdot\bz)^2}+\mb{I},
\end{equation}
with the superscript $T$ denoting transposition and $\mb{I}$ denoting the $N\times N$ identity matrix; the expression follows from cancelling the term $-\sqrt{\beta}\bx\cdot\bz$ by using \eqref{eq:peakvector} in the first term of the expansion \eqref{eq:logexpansionw}, leaving the second term in the expansion to form $\mb{H}$.

Because of the multivariate centered Gaussian in the integral and the symmetry of $\mcl{I}(\beta)$, all antisymetric terms vanish (these include the orders $\beta^{-1/2}$ and $\beta^{-3/2}$), so expanding up to terms of order $\beta^{-1}$ gives
\begin{align}
I_1=&\int_{\mcl{I}(\beta)}\text{e}^{-\bx^T\mb{H}\bx/2}\Big[1-\frac{1}{4\beta}\sum_{\bzeta\in R_+}k(\bzeta)\Big(\frac{\bzeta\cdot\bx}{\bzeta\cdot\bz}\Big)^4+\frac{1}{\beta}\Big(\frac{\balpha\cdot\bx}{\balpha\cdot\bz}\Big)^2\notag\\
&\qquad+\frac{1}{2\beta}\Big(2\frac{\balpha\cdot\bx}{\balpha\cdot\bz}-\frac{1}{3}\sum_{\bzeta\in R_+}k(\bzeta)\Big[\frac{\bzeta\cdot\bx}{\bzeta\cdot\bz}\Big]^3\Big)^2+O(\beta^{-2})\Big]\intd^N \bx.
\end{align}

Let us focus on the region $\bx\in C_W(\sqrt{\beta}\bz)\setminus\mcl{I}(\beta)$ now. We can rewrite the integrand in \eqref{eq:mainintegral} as one exponential and simplify it by using the inequality $x>\log(1+x)$ as well as \eqref{eq:peakvector},
\begin{align}
\exp\Big\{&-\frac{\|\bx\|^2}{2}-\sqrt{\beta}\bz\cdot\bx+\beta\sum_{\bzeta\in R_+}k(\bzeta)\log\Big(1+\frac{\bzeta\cdot\bx}{\sqrt{\beta}\bzeta\cdot\bz}\Big)-2\log\Big(1+\frac{\balpha\cdot\bx}{\sqrt{\beta}\balpha\cdot\bz}\Big)\Big\}\notag\\
&\leq \exp\Big\{-\frac{1}{2}\Big\|\bx+\frac{2\balpha}{\sqrt{\beta}\balpha\cdot\bz}\Big\|^2+\frac{2\|\balpha\|^2}{\beta(\balpha\cdot\bz)^2}\Big\}.
\end{align}
With this, we can write
\begin{align}
I_2&:=\int_{C_W(\sqrt{\beta}\bz)\setminus\mcl{I}(\beta)}\frac{\text{e}^{-\|\bx\|^2/2-\sqrt{\beta}\bx\cdot\bz}}{[1+\frac{\balpha\cdot\bx}{\sqrt{\beta}\balpha\cdot\bz}]^2}\prod_{\bzeta\in R_+}\Big(1+\frac{\bzeta\cdot\bx}{\sqrt{\beta}\bzeta\cdot\bz}\Big)^{\beta k(\bzeta)}\intd^N \bx\notag\\
&\leq\int_{C_W(\sqrt{\beta}\bz)\setminus\mcl{I}(\beta)}\exp\Big\{-\frac{1}{2}\Big\|\bx+\frac{2\balpha}{\sqrt{\beta}\balpha\cdot\bz}\Big\|^2+\frac{2\|\balpha\|^2}{\beta(\balpha\cdot\bz)^2}\Big\}\intd^N \bx\notag\\
&\leq\int_{\substack{\bx:\|\bx\|>\sqrt{\beta\gamma}r^\star\\ \bx\cdot\bz>0}}\exp\Big\{-\frac{1}{2}\Big\|\bx+\frac{2\balpha}{\sqrt{\beta}\balpha\cdot\bz}\Big\|^2+\frac{2\|\balpha\|^2}{\beta(\balpha\cdot\bz)^2}\Big\}\intd^N \bx,
\end{align}
with $r^\star$ as in the statement. Finally, there exists a constant $C>0$ such that
\begin{equation}
I_2\leq \frac{C \varphi_N}{2}\int_{\sqrt{\beta\gamma}r^\star}^\infty \text{e}^{-x^2/2}x^{N-1}\intd x=O(\text{e}^{-\beta\gamma r^\star/2})=O(\mnm{e}^{-M}),
\end{equation}
where $\beta>2M/\gamma r^\star$. $\varphi_N$ represents the total solid angle in $N$ dimensions. In a similar manner, we can find the asymptotics of the integral
\begin{align}
\frac{c_\beta}{|W|}&=\int_{C_W}\text{e}^{-\|\bx\|^2/2}w_\beta(\bx)\intd^N\bx\notag\\
&=\Big(\frac{\beta}{\text{e}}\Big)^{\beta\gamma/2}w_\beta(\bz)\int_{C_W(\sqrt{\beta}\bz)}\text{e}^{-\|\bx\|^2/2-\sqrt{\beta}\bx\cdot\bz}\prod_{\bzeta\in R_+}\Big(1+\frac{\bzeta\cdot\bx}{\sqrt{\beta}\bzeta\cdot\bz}\Big)^{\beta k(\bzeta)}\intd^N \bx\notag\\
&=\Big(\frac{\beta}{\text{e}}\Big)^{\beta\gamma/2}w_\beta(\bz)\ (I_3+I_4),
\end{align}
with $I_3$ and $I_4$ denoting the integrals in regions $\mcl{I}(\beta)$ and $C_W(\sqrt{\beta}\bz)\setminus\mcl{I}(\beta)$, respectively. Then, we see that
\begin{equation}
I_3=\int_{\mcl{I}(\beta)}\text{e}^{-\bx^T\mb{H}\bx/2}\Big[1-\frac{1}{4\beta}\sum_{\bzeta\in R_+}k(\bzeta)\Big(\frac{\bzeta\cdot\bx}{\bzeta\cdot\bz}\Big)^4+O(\beta^{-2})\Big]\intd^N\bx
\end{equation}
and
\begin{equation}
I_4\leq \frac{\varphi_N}{2}\int_{\sqrt{\beta\gamma}r^\star}^\infty \text{e}^{-x^2/2}x^{N-1}\intd x=O(\text{e}^{-\beta\gamma r^\star/2}).
\end{equation}
Let us set
\[\tilde{c}_\beta:=\int_{\mcl{I}(\beta)}\text{e}^{-\bx^T\mb{H}\bx/2}\intd^N\bx>0\]
in order to write
\begin{align}
\lambda_{\balpha}(\beta)&=\frac{\|\balpha\|^2k(\balpha)}{4(\balpha\cdot\bz)^2}\frac{I_1+I_2}{I_3+I_4}\notag\\
&=\frac{\|\balpha\|^2k(\balpha)}{4(\balpha\cdot\bz)^2}\Big[1+\frac{1}{\beta\tilde{c}_\beta}\int_{\mcl{I}(\beta)}\text{e}^{-\bx^T\mb{H}\bx/2}\Big[\Big(\frac{\balpha\cdot\bx}{\balpha\cdot\bz}\Big)^2\notag\\
&\qquad+\frac{1}{2}\Big(2\frac{\balpha\cdot\bx}{\balpha\cdot\bz}-\frac{1}{3}\sum_{\bzeta\in R_+}k(\bzeta)\Big[\frac{\bzeta\cdot\bx}{\bzeta\cdot\bz}\Big]^3\Big)^2\Big]\intd^N \bx+O(\beta^{-2})+O(\mnm{e}^{-M})\Big]\notag\\
&=\frac{\|\balpha\|^2k(\balpha)}{4(\balpha\cdot\bz)^2}\Big[1+\frac{\tilde{C}_\beta(\balpha)}{\beta}+O(\beta^{-2})+O(\mnm{e}^{-M})\Big].
\end{align}
The final line is obtained by expanding the denominator in powers of $\beta$. The constant $\tilde{C}_\beta(\balpha)$ is given by
\begin{align}
\tilde{C}_\beta(\balpha):=\frac{1}{\tilde{c}_\beta}\int_{\mcl{I}(\beta)}&\text{e}^{-\bx^T\mb{H}\bx/2}\Big[\Big(\frac{\balpha\cdot\bx}{\balpha\cdot\bz}\Big)^2\notag\\
&+\frac{1}{2}\Big(2\frac{\balpha\cdot\bx}{\balpha\cdot\bz}-\frac{1}{3}\sum_{\bzeta\in R_+}k(\bzeta)\Big[\frac{\bzeta\cdot\bx}{\bzeta\cdot\bz}\Big]^3\Big)^2\Big]\intd^N \bx>0,
\end{align}
and finally
\begin{align}
\tilde{C}(\balpha):=\sqrt{\frac{\det(\mb{H})}{(2\pi)^N}}\int_{\mbb{R}^N}&\text{e}^{-\bx^T\mb{H}\bx/2}\Big[\Big(\frac{\balpha\cdot\bx}{\balpha\cdot\bz}\Big)^2\notag\\
&+\frac{1}{2}\Big(2\frac{\balpha\cdot\bx}{\balpha\cdot\bz}-\frac{1}{3}\sum_{\bzeta\in R_+}k(\bzeta)\Big[\frac{\bzeta\cdot\bx}{\bzeta\cdot\bz}\Big]^3\Big)^2\Big]\intd^N \bx +O(\mnm{e}^{-M}).
\end{align}
The statement follows.
\end{proof}

The large-$\beta$ behavior of the transition coefficients gives us insight on the behavior of the relaxation exponents in the same regime: when $\beta$ is large but finite the relaxation exponents are larger in magnitude than their value when $\beta\to\infty$.

\begin{proof}[Proof of  Theorem~\ref{th:perturbation}]
It suffices to perform a first-order perturbation on the eigenfunctions and eigenvalues of the matrix $M_{\tau\nu}=M_{\tau\nu}(\beta)$. We assume that $\beta>C(M)=2M/\gamma r^\star$ in order to make use of Lemma~\ref{lm:coefficientasymptotics}, and we denote the eigenvalues of $\lim_{\beta\to\infty}M_{\tau\nu}(\beta)$ by $\{r_i^{(0)}\}_i$ in descending order. First, we note that, since \eqref{eq:transitionratecoefficientssum} holds for all $t$ and $\beta>1$, we can evaluate it at $\bx_0=\mb{0}$ and $t=1$, and use Lemma~\ref{lm:coefficientasymptotics} and the expansion
\begin{equation}
\frac{\beta|R_+|}{4(\beta-1)}=\frac{|R_+|}{4}\sum_{n=0}^{\infty}\beta^{-n}
\end{equation}
to equate the coefficients of $\beta^{-1}$ and obtain (abbreviating the superscript of the peak vector $\bz^{(R)}$)
\begin{equation}\label{eq:sumctilde}
\sum_{\balpha\in R_+}\frac{\|\balpha\|^2\tilde{C}(\balpha)}{4(\balpha\cdot\bz)^2}=\frac{|R_+|}{4}.
\end{equation}
Let us write the eigenvectors of $M_{\tau\nu}(\beta)$ as
\begin{equation}
\phi_i(\beta;\tau)=\sum_{n=0}^{\infty}\beta^{-n}\phi_i^{(n)}(\tau).
\end{equation}
They form an orthonormal basis because $M_{\tau\nu}(\beta)$ is a symmetric matrix, and we write their respective eigenvalues as
\begin{equation}
r_i(\beta)=\sum_{n=0}^{\infty}\beta^{-n}r_i^{(n)}.
\end{equation}
We note that the $\{\phi_i^{(0)}(\tau)\}_{i=1}^{|W|}$ form an orthonormal basis themselves, as they are the eigenvector basis that diagonalizes $M_{\tau\nu}(\beta)$ when $\beta\to\infty$.
Inserting these expressions and the result of Lemma~\ref{lm:coefficientasymptotics} into the eigenvalue equation
\[\sum_{\nu\in W}M_{\tau\nu}(\beta)\phi_i(\beta;\nu)=r_i(\beta)\phi_i(\beta;\nu),\]
and equating the coefficients of $\beta^{-1}$ yields
\begin{align}
\sum_{\balpha\in R_+}\frac{\|\balpha\|^2}{4(\balpha\cdot\bz)^2}&[\tilde{C}(\balpha)\phi_i^{(0)}(\tau\sigma_{\balpha})+\phi_i^{(1)}(\tau\sigma_{\balpha})]-\frac{|R_+|}{4}[\phi_i^{(0)}(\tau)+\phi_i^{(1)}(\tau)]\notag\\
&=r_i^{(0)}\phi_i^{(1)}(\tau)+r_i^{(1)}\phi_i^{(0)}(\tau).\label{eq:perturbationequation}
\end{align}
We solve for $r_i^{(1)}$ by multiplying $\sum_{\tau\in W}\phi_i^{(0)}(\tau)$ from the left, as well as making use of the relation
\begin{align}
\sum_{\balpha\in R_+}&\frac{\|\balpha\|^2}{4(\balpha\cdot\bz)^2}\sum_{\tau\in W}\phi_i^{(0)}(\tau)\phi_i^{(1)}(\tau\sigma_{\balpha})-\frac{|R_+|}{4}\sum_{\tau\in W}\phi_i^{(0)}(\tau)\phi_i^{(1)}(\tau)\notag\\
&=\sum_{\balpha\in R_+}\frac{\|\balpha\|^2}{4(\balpha\cdot\bz)^2}\sum_{\tau\in W}\phi_i^{(0)}(\tau\sigma_{\balpha})\phi_i^{(1)}(\tau)-\frac{|R_+|}{4}\sum_{\tau\in W}\phi_i^{(0)}(\tau)\phi_i^{(1)}(\tau)\notag\\
&=\sum_{\tau\in W}\phi_i^{(1)}(\tau)\Big[\sum_{\balpha\in R_+}\frac{\|\balpha\|^2}{4(\balpha\cdot\bz)^2}\sum_{\tau\in W}\phi_i^{(0)}(\tau\sigma_{\balpha})-\frac{|R_+|}{4}\phi_i^{(0)}(\tau)\Big]\notag\\
&=r_i^{(0)}\sum_{\tau\in W}\phi_i^{(1)}(\tau)\phi_i^{(0)}(\tau).
\end{align}
(The second equality is obtained by performing the substitution $\tau\sigma_{\balpha}\to\tau$.) Then, \eqref{eq:perturbationequation} becomes
\begin{equation}
r_i^{(1)}=\sum_{\balpha\in R_+}\frac{\|\balpha\|^2\tilde{C}(\balpha)}{4(\balpha\cdot\bz)^2}\sum_{\tau\in W}\phi_i^{(0)}(\tau)\phi_i^{(0)}(\tau\sigma_{\balpha})-\frac{|R_+|}{4}.
\end{equation}
We use the orthonormality of the $\{\phi_i^{(0)}\}_{i=1}^{|W|}$ and \eqref{eq:sumctilde} to complete the square and obtain
\begin{equation}
r_i^{(1)}=-\frac{1}{8}\sum_{\balpha\in R_+}\frac{\|\balpha\|^2\tilde{C}(\balpha)}{(\balpha\cdot\bz)^2}\sum_{\tau\in W}[\phi_i^{(0)}(\tau)-\phi_i^{(0)}(\tau\sigma_{\balpha})]^2.
\end{equation}
Because $\tilde{C}(\balpha)$ and the rest of the terms in the sum are positive, it follows that $r_i^{(1)}<0$.
\end{proof}

This result gives formal footing to the intuition that the relaxation should be faster at finite $\beta$, and that in the large-$\beta$ regime the lower bound is given by the least negative nonzero eigenvalue of $M_{\tau\nu}(\beta)$ at $\beta\to\infty$.

\section{Relaxation for the case $A_{N-1}$}\label{sec:relaxationA}

The case $A_{N-1}$ is of particular interest because of its relationship with the Dyson model. Here, we consider the relaxation of the jump process in terms of the least negative eigenvalue and we draw a connection to PF spin chains \cite{Polychronakos92, Frahm93} in the freezing regime, namely, the limit $\beta\to\infty$.

Let us examine the freezing regime first. It is well known that the peak vector (or Fekete set \cite{Fekete23}) for $A_{N-1}$ is given by the zeroes of the $N$-th Hermite polynomial $H_N(x)$, $\mb{z}=(z_1,\ldots, z_N)$, in ascending order. Here, the Hermite polynomials are defined as the family of polynomials $\{H_n(x)\}_{n=0}^\infty$ that are orthogonal in $\mbb{R}$ with respect to $\exp(-x^2)$ \cite{Szego}. We denote the roots of the positive subsystem by $\balpha_{j,i}=\mb{e}_j-\mb{e}_i$, $i<j$. Then, the jump rates are given by
\begin{equation}
\lambda_{\beta\to\infty}(1,\balpha_{j,i}|\mb{0})=\frac{1}{2(z_j-z_i)^2},
\end{equation}
and, adopting the notation in Theorem~\ref{th:perturbation} where the superscript $(0)$ indicates the limit $\beta\to\infty$ of the corresponding quantity, the action of the transition matrix in the master equation on an arbitrary function $f(\tau)$ becomes
\begin{equation}\label{eq:eigenvaluefreezingA}
\sum_{\nu\in S_N}M_{\tau\nu}^{(0)}f(\nu)=\frac{1}{2}\sum_{1\leq i<j\leq N}\frac{f(\tau\sigma_{j,i})}{(z_j-z_i)^2}-\frac{N (N-1)}{8}f(\tau).
\end{equation}
Here, we have set $\sigma_{j,i}:=\sigma_{\balpha_{j,i}}$ for brevity. As in Lemma~\ref{lm:minimumeigvl}, let us define the matrix $\tilde{M}_{\tau\nu}^{(0)}$ by
\begin{equation}
\sum_{\nu\in S_N}\tilde{M}_{\tau\nu}^{(0)}f(\nu)=\frac{1}{2}\sum_{1\leq i<j\leq N}\frac{f(\tau\sigma_{j,i})}{(z_j-z_i)^2}.
\end{equation}
This expression is nothing but the Hamiltonian of a PF spin chain with $N$ spins multiplied by one half. To see this, we introduce the following bra-ket notation: given a family of $N$-dimensional Hilbert spaces $\{\mcl{H}_i\}_{i=1}^N$ denote by $\{|n\rangle_i\}_{n=1}^N$ each of their respective canonical orthonormal bases, that is, $|n\rangle_i\in\mcl{H}_i$ with ${}_i\langle m|n\rangle_i=\delta_{m,n}$ and $\sum_{n=1}^N|n\rangle_i\langle n|_i=\mb{I}_i$. Consider the direct product of the Hilbert spaces $\bigotimes_{i=1}^N \mcl{H}_i$, and define for every multi-index $\mb{n}\in\{1,\ldots,N\}^N$ the vector
\begin{equation}
|\mb{n}\rangle:=|n_1\rangle_1\otimes\cdots\otimes|n_N\rangle_N\in\bigotimes_{i=1}^N \mcl{H}_i.
\end{equation}
We will restrict ourselves to the subspace $\tilde{\mcl{H}}\subset\bigotimes_{i=1}^N \mcl{H}_i$ spanned by the vectors $|\mb{n}\rangle$ such that there exists a permutation $\rho\in S_N$ for which $n_i=\rho(i)$. In that case, we write $|\rho\rangle:=|\mb{n}\rangle$. Using this notation, we can write functions $f:S_N\to\mbb{R}$ as vectors in $\tilde{\mcl{H}}$, namely
\[f(\rho)=\langle \rho|f \rangle\quad\Leftrightarrow\quad |f\rangle = \sum_{\nu\in S_N}f(\nu)|\nu\rangle.\]
Similarly, linear operators acting on $f$ can be expressed as operators acting on $|\rho\rangle$; for $\tau,\nu\in S_N$ and a generic matrix $A_{\tau\nu}$, we define $\hat{A}:\tilde{\mcl{H}}\to\tilde{\mcl{H}}$ by $\langle\tau|\hat{A}|\nu\rangle:=A_{\tau\nu}$ so that
\[\sum_{\nu\in S_N}A_{\tau\nu}f(\nu)=\langle\tau|\hat{A}|f\rangle.\]
Then, we can leave the matrix $\tilde{M}_{\tau\nu}^{(0)}$ aside and instead focus on its corresponding operator $\hat{M}$,
\begin{equation*}
\hat{M}=\frac{1}{2}\sum_{1\leq i<j\leq N}\frac{1}{(z_j-z_i)^2}\sum_{\tau\in S_N}|\tau\sigma_{j,i}\rangle\langle\tau|,
\end{equation*}
which naturally has the same spectrum and also maps $\tilde{\mcl{H}}$ onto itself. The action of the operator in the sum over $\tau$ is that of exchanging the $i$-th and $j$-th components of $|\tau\rangle$, namely $\tau(i)$ and $\tau(j)$, so we can rewrite this as simply
\begin{equation}
\hat{M}=\frac{1}{2}\sum_{1\leq i<j\leq N}\frac{\hat{P}_{j,i}}{(z_j-z_i)^2}.
\end{equation}
This is the usual way in which the PF spin chain Hamiltonian is written. The spectrum of this operator is calculated by regarding each of the subspaces $\mcl{H}_i$ as an $SU(N)$ spin and introducing the $N^2-1$ matrices of the $\mfk{su}(N)$ algebra,
\begin{equation}
\hat{J}^{(j,l)}_i=
\begin{cases}
\frac{1}{2}(|j\rangle_i\langle l|_i+|l\rangle_i\langle j|_i), & \text{ if }1\leq j<l\leq N,\\
\frac{\sqrt{-1}}{2}(|l\rangle_i\langle j|_i-|j\rangle_i\langle l|_i), & \text{ if }1\leq l<j\leq N,\\
\frac{1}{\sqrt{2j(j+1)}}(\sum_{m=1}^j|m\rangle_i\langle m|_i-j|j+1\rangle_i\langle j+1|_i), & \text{ if }1\leq j=l\leq N-1,\\
0,&\text{ if }j=l=N.
\end{cases}
\end{equation}
Let us keep in mind that we are using the following shorthand notation
\[\hat{J}^{(j,l)}_i=\hat{I}_1\otimes\cdots\hat{I}_{i-1}\otimes\hat{J}^{(j,l)}_i\otimes\hat{I}_{i+1}\otimes\cdots\hat{I}_{N}\]
to let the subscript indicate the subspace on which $\hat{J}^{(j,l)}_i$ operates, with $\hat{I}_n$ denoting the identity operator acting on $\mcl{H}_n$; then
\[\hat{J}^{(j,l)}_i|\rho\rangle=|\rho(1)\rangle\otimes\cdots\otimes\hat{J}^{(j,l)}_i|\rho(i)\rangle_i\otimes\cdots\otimes|\rho(N)\rangle_N.\]
Finally, we set the convention that all subspaces of index different from any subscript in an operator are to be left intact in general. Then, it can be shown that $\hat{P}_{m,n}$ can be expressed as (see Appendix~\ref{sec:appendix})
\begin{equation}\label{eq:exchangeop}
\hat{P}_{m,n}=\frac{1}{N}\hat{I}_m\otimes\hat{I}_n+2\sum_{j,l}\hat{J}^{(j,l)}_m\otimes\hat{J}^{(j,l)}_n.
\end{equation}
Finally, recall that
\begin{align}\label{eq:propertiesofj}
\Tr[\hat{J}^{(j_1,l_1)}_i,\hat{J}^{(j_2,l_2)}_i]&=\frac{1}{2}\delta_{j_1,j_2}\delta_{l_1,l_2}\quad\text{and}\notag\\
[\hat{J}^{(j_1,l_1)}_i,\hat{J}^{(j_2,l_2)}_i]&=\sum_{j_3,l_3}f^{[(j_1,l_1),(j_2,l_2),(j_3,l_3)]}\hat{J}^{(j_3,l_3)}_i,
\end{align}
where $[\cdot,\cdot]$ denotes the commutator and the $f^{[(j_1,l_1),(j_2,l_2),(j_3,l_3)]}$ denote the structure constants.
With all the definitions in place, we can define ladder operators in order to produce the spectrum of $\hat{M}$. Let us define operators $\hat{K}^{(j,l)},\hat{L}^{(j,l)}$ that map $\bigotimes_{i=1}^N \mcl{H}_i$ onto itself by
\begin{align}
\hat{K}^{(j,l)}&:=\sum_{m=1}^N z_m\hat{J}^{(j,l)}_m\quad\text{and}\\
\hat{L}^{(j,l)}&:=\sum_{\substack{(j_1,l_1)\\(j_2,l_2)}}\sum_{1\leq m\neq n\leq N}\frac{f^{[(j,l),(j_1,l_1),(j_2,l_2)]}}{z_m-z_n}\hat{J}^{(j_1,l_1)}_m\hat{J}^{(j_2,l_2)}_n.
\end{align}
Do note that $\hat{K}^{(j,l)}$ and $\hat{L}^{(j,l)}$ do not map $\tilde{\mcl{H}}$ onto itself in general. We check in Appendix~\ref{sec:appendix} that
\begin{equation}\label{eq:klmcommutators}
[\hat{M},\hat{K}^{(j,l)}]=\frac{1}{2}\hat{L}^{(j,l)}\quad\text{and}\quad[\hat{M},\hat{L}^{(j,l)}]=\frac{1}{2}\hat{K}^{(j,l)}.
\end{equation}
Then, for every eigenvector $|\phi_m\rangle$ with eigenvalue $\tilde{r}^{(0)}_m$, it follows that 
\begin{equation}
\hat{M}(\hat{K}^{(j,l)}\pm\hat{L}^{(j,l)})|\phi_m\rangle=(\tilde{r}^{(0)}_m\pm1/2)(\hat{K}^{(j,l)}\pm\hat{L}^{(j,l)})|\phi_m\rangle.
\end{equation}
Finally, we recover the eigenvalues of $\tilde{M}_{\tau\nu}^{(0)}$ from \eqref{eq:eigenvaluefreezingA} by writing
\begin{equation}
r_i^{(0)}=\tilde{r}_i^{(0)}-\frac{N (N-1)}{8}.
\end{equation}
By \eqref{eq:negsemidef}, we know that the maximum eigenvalue of $M_{\tau\nu}^{(0)}$ is zero, so we can use an annihilation operator $\hat{K}^{(j,l)}-\hat{L}^{(j,l)}$ to find that the least negative nonzero eigenvalue is $-1/2$. However, we cannot use an arbitrary annihilation operator because we require that the result of operating $\hat{K}^{(j,l)}-\hat{L}^{(j,l)}$ on the eigenvector $\sum_{\rho\in S_N}(N!)^{-1/2}|\rho\rangle$ remain in $\tilde{\mcl{H}}$. We show in the Appendix that this requirement is only satisfied by annihilation operators such that $j=l$, meaning that there exist exactly $N-1$ eigenvectors with eigenvalue $-1/2$.

From the result in the previous section, we find that the scaling exponent of relaxation is given by
\begin{equation}
r_{1, A}=-\Big(\frac{1}{2}+\frac{|r_{1, A}^{(1)}|}{\beta}\Big)+O(\beta^{-2})
\end{equation}
when $\beta>C(M)$ as in Theorem~\ref{th:perturbation}, provided $|r_{1,A}|<1$.

\section{Concluding remarks}\label{sec:remarks}

We obtained the master equation that governs the dynamics of Dunkl jump processes when $\beta>1$ and all the multiplicities $k(\balpha)\geq 1$. This parameter condition was critical in the description of the jump processes as Poisson random walks on the Weyl group $W$, as it provided the basis for the proof of Lemma~\ref{jumplemma}. However, the dynamics of the process when $\beta<1$ must be handled with greater care due to the fact that the rate functions in \eqref{eq:jumprates} diverge and our approach breaks down.

For the cases $A_{N-1}$ and $B_N$, which have a corresponding stochastic particle system representation, we reported a phase transition that appears as a singularity for $\beta_c=1$ in the jump rate per particle in the bulk scaling limit $N\to\infty$ with $t=N$. While it is well-known that $\beta_c=1$ is the value of $\beta$ that separates the colliding ($\beta<\beta_c$) and non-colliding ($\beta>\beta_c$) characteristics of these particle systems, no other physical insight had been given until now. Given the form found in Theorem~\ref{th:phasetransition}, it is clear that the critical exponent in the ordered ($\beta>\beta_c$) phase is equal to one. However, it is unknown whether a similar result can be obtained in the disordered ($\beta<\beta_c$) phase. In addition, the bulk scaling limit was taken at the level of the jump rate per particle directly, so we do not have any further specific information about the jump process behavior in the limit $N\to\infty$; this is a problem that warrants further investigation.

We examined the $\beta$-dependence of the relaxation exponent and we found that the initial configuration $\bx_0$ plays an important role in the long-time dynamics. In particular, we found that the relaxation exponent magnitude has an upper bound of 1 which is imposed by the effect of $\bx_0$ on the jump rates. Moreover, we found that this exponent changes at large values of $\beta$ whenever the case $\bx_0=\mb{0}$ at $\beta\to\infty$ yields an exponent of magnitude less than one. It is interesting to note that it was reported in \cite{AndrausMiyashita15} that the relaxation asymptotics due to the jumps in a Dunkl process were of order $t^{-1/2}$, which is a correction derived from the power series expansion of the Dunkl kernel. We have shown here that these asymptotics are a worst-case scenario, as the fastest jump process relaxation behaves like $t^{-1}$ according to Theorem~\ref{th:relaxation}. This leads us to suspect that the relaxation exponent satisfies $-1\leq r_1(\beta)\leq-1/2$ in general. This seems consistent with the case $A_{N-1}$, where we found that $r_{1,A}(\beta)\uparrow-1/2$ when $\beta\to\infty$. While we do expect $r_1(\beta)$ to be a non-increasing function of $\beta$ in general, we do not have proof of this yet.

Three open questions remain after the present work. An important one is describing the jump process when $0<\beta<1$, as in that regime the description given here breaks down. This implies that the jump process must be approached using a different technique, and in particular, that taking expectation values with respect to the law of the underlying radial process must be avoided.

Another open problem is calculating each of the jump rates as functions of $\beta>1$. Computing the total jump rate for $\bx_0=\mb{0}$ was straightforward because it is the integral of a $W$-invariant function. However, the individual jump rates are integrals of non-$W$-invariant functions, which makes their calculation substantially more difficult, even when $\bx_0=\mb{0}$. We suspect that there may be indirect methods for calculating them, and we would like to make progress in this direction.

Finally, we mentioned that in general the relaxation exponent is decreasing in absolute value for sufficiently large $\beta$, and that in particular it falls toward 1/2 in the case $A_{N-1}$. This decreasing behavior appeared from the eigenvalue problem for the case $\bx_0=\mb{0}$; we suspect that these eigenvalues are decreasing for $\beta>1$. If we could prove this, it would imply the existence of a cross-over dictated by a well-defined value of $\beta$ under which the relaxation exponent is 1 and above which the exponent is less than 1. We are currently considering this problem, and hope to report on it in the near future.

\section*{Acknowledgments}
The author would like to thank Makoto Katori for helpful discussions and comments, and an anonymous referee for several recommendations that led to improvements in the present work. The author is supported by JSPS KAKENHI Grant Number JP19K14617. 

\appendix

\section{}\label{sec:appendix}

First, we provide a proof of \eqref{eq:exchangeop}, followed by a derivation of the relationships \eqref{eq:klmcommutators}. After that, we show that the only ladder operators which map $\tilde{\mcl{H}}$ into itself must be of the form $\hat{K}^{(j,j)}\pm\hat{L}^{(j,j)}$, $1\leq j\leq N-1$.

Without loss of generality, we can consider the case where $\mcl{H}=\mcl{H}_1\otimes\mcl{H}_2$, and each of the  $\mcl{H}_i$ is of dimension $N$. Then, we consider the action of
\begin{equation}
\hat{P}_{1,2}=\frac{1}{N}\hat{I}_1\otimes\hat{I}_2+2\sum_{1\leq j,l\leq N}\hat{J}^{(j,l)}_1\otimes\hat{J}^{(j,l)}_2
\end{equation}
on the vector
\begin{equation}
|m,n\rangle:=|m\rangle_1\otimes|n\rangle_2.
\end{equation}
Since we are dealing with permutations, we can also impose $m<n$. Henceforth, we abbreviate the direct product symbols for simplicity. Let us denote by $\mbb{I}_{\eta}$ the indicator function, that is, $\mbb{I}_{\eta}=1$ if the expression $\eta$ is true, and zero if it is false. Then, for $j=l$ we have
\begin{align}
\sum_{j=1}^{N-1}\hat{J}^{(j,j)}_1\hat{J}^{(j,j)}_2|m,n\rangle&=\sum_{j=1}^{N-1}\frac{1}{2j(j+1)} [\mbb{I}_{j\geq n}-j\delta_{j+1,n}\mbb{I}_{j\geq m}]|m,n\rangle\notag\\
&=\frac{1}{2}\Big[\sum_{j=n}^{N-1}\frac{1}{j(j+1)}-\frac{1}{n}\Big]|m,n\rangle\notag\\
&=\frac{1}{2}\Big[\sum_{j=n+1}^{N-1}\frac{1}{j(j+1)}-\frac{1}{n+1}\Big]|m,n\rangle\notag\\
&\qquad\vdots\notag\\
&=\frac{1}{2}\Big[\frac{1}{N(N-1)}-\frac{1}{N-1}\Big]|m,n\rangle=-\frac{1}{2N}|m,n\rangle.
\end{align}
Similarly, for $j<l$ we get
\begin{align}
\sum_{1\leq j<l\leq N}\hat{J}^{(j,l)}_1\hat{J}^{(j,l)}_2|m,n\rangle=&\frac{1}{4}\sum_{1\leq j<l\leq N}(|j,l\rangle\langle l,j|+|l,j\rangle\langle j,l|)|m,n\rangle\notag\\
=&\frac{1}{4}|n,m\rangle,
\end{align}
and the case $j>l$ gives
\begin{align}
\sum_{1\leq l<j\leq N}\hat{J}^{(j,l)}_1\hat{J}^{(j,l)}_2|m,n\rangle=&-\frac{1}{4}\sum_{1\leq l<j\leq N}(-|l,j\rangle\langle j,l|-|j,l\rangle\langle l,j|)|m,n\rangle\notag\\
=&\frac{1}{4}|n,m\rangle.
\end{align}
Adding all of the expressions gives
\begin{align}
\hat{P}_{1,2}|m,n\rangle&=\frac{1}{N}|m,n\rangle+2\Big[\frac{1}{4}|n,m\rangle+\frac{1}{4}|n,m\rangle-\frac{1}{2N}|m,n\rangle\Big]=|n,m\rangle,
\end{align}
as desired.

The first expression in \eqref{eq:klmcommutators} is derived through a direct calculation. Inserting the definitions of $\hat{M}$ and $\hat{K}^{(j,l)}$ into the commutator gives
\begin{align}
[\hat{M},\hat{K}^{(j,l)}]=&\frac{1}{2}\sum_{m=1}^N\sum_{1\leq i<k\leq N}\frac{z_m}{(z_k-z_i)^2} [\hat{P}_{k,i},\hat{J}^{(j,l)}_m]\notag\\
=&\frac{1}{2}\sum_{m=1}^N\sum_{1\leq i<k\leq N}\frac{z_m}{(z_k-z_i)^2}\Big(\frac{1}{N}[\hat{I}_k\hat{I}_i,\hat{J}^{(j,l)}_m]\notag\\
&\qquad\qquad\qquad+2\sum_{j_1,l_1}[\hat{J}_k^{(j_1,l_1)}\hat{J}_i^{(j_1,l_1)},\hat{J}_m^{(j,l)}]\Big)\notag\\
=&\sum_{m=1}^N\sum_{1\leq i<k\leq N}\sum_{j_1,l_1}\frac{z_m}{(z_k-z_i)^2}\Big(\delta_{m,i}\hat{J}_k^{(j_1,l_1)}[\hat{J}_i^{(j_1,l_1)},\hat{J}_i^{(j,l)}]\notag\\
&\qquad\qquad\qquad+\delta_{m,k}\hat{J}_i^{(j_1,l_1)}[\hat{J}_k^{(j_1,l_1)},\hat{J}_k^{(j,l)}]\Big)\notag\\
=&\sum_{1\leq i<k\leq N}\frac{z_i}{(z_k-z_i)^2}\sum_{j_1,l_1}\hat{J}_k^{(j_1,l_1)}[\hat{J}_i^{(j_1,l_1)},\hat{J}_i^{(j,l)}]\notag\\
&+\sum_{1\leq i<k\leq N}\frac{z_k}{(z_k-z_i)^2}\sum_{j_1,l_1}\hat{J}_i^{(j_1,l_1)}[\hat{J}_k^{(j_1,l_1)},\hat{J}_k^{(j,l)}]\notag\\
=&\sum_{1\leq i<k\leq N}\frac{z_i}{(z_k-z_i)^2}\sum_{j_1,l_1}\sum_{j_2,l_2}f^{[(j_1,l_1),(j,l),(j_2,l_2)]}\hat{J}_k^{(j_1,l_1)}\hat{J}^{(j_2,l_2)}_i\notag\\
&+\sum_{1\leq i<k\leq N}\frac{z_k}{(z_k-z_i)^2}\sum_{j_1,l_1}\sum_{j_2,l_2}f^{[(j_1,l_1),(j,l),(j_2,l_2)]}\hat{J}_i^{(j_1,l_1)}\hat{J}_k^{(j_2,l_2)}\notag\\
=&\sum_{1\leq i<k\leq N}\frac{z_k-z_i}{(z_k-z_i)^2}\sum_{j_1,l_1}\sum_{j_2,l_2}f^{[(j_1,l_1),(j,l),(j_2,l_2)]}\hat{J}_i^{(j_1,l_1)}\hat{J}_k^{(j_2,l_2)}\notag\\
=&\frac{1}{2}\sum_{1\leq i\neq k\leq N}\frac{1}{z_k-z_i}\sum_{j_1,l_1}\sum_{j_2,l_2}f^{[(j_1,l_1),(j,l),(j_2,l_2)]}\hat{J}_i^{(j_1,l_1)}\hat{J}_k^{(j_2,l_2)}\notag\\
=&\frac{1}{2}\sum_{1\leq i\neq k\leq N}\frac{1}{z_i-z_k}\sum_{j_1,l_1}\sum_{j_2,l_2}f^{[(j,l),(j_1,l_1),(j_2,l_2)]}\hat{J}_i^{(j_1,l_1)}\hat{J}_k^{(j_2,l_2)}\notag\\
=&\frac{1}{2}\hat{L}^{(j,l)}.
\end{align}
For the fifth equality, we have used \eqref{eq:propertiesofj}, for the sixth, we have used the antisymmetry of the structure constants after exchanging the indices $(j_1,l_1)\leftrightarrow(j_1,l_2)$ in the first term. To obtain the seventh equality, we note that the sum is invariant when we exchange $i$ and $k$, and we end the calculation using the antisymmetry of the structure constants one last time. The second expression in \eqref{eq:klmcommutators} is obtained in a similar manner; terms involving $z_i$, $z_k$ and $z_m$ ($i\neq k\neq m$) appear, but they vanish because the resulting sum is antisymmetric in two of the indices, adding up to zero. The remaining terms involve only coefficients of the form $(z_i-z_k)^{-3}$, and making use of the following relation
\begin{equation}
z_i=\sum_{\substack{1\leq k\leq N\\k:k\neq i}}\frac{2}{(z_i-z_k)^3}
\end{equation}
for the zeroes of $H_N$ yields the desired result. This last relation is obtained by taking the second derivative of the Hermite differential equation
\begin{equation}
H_N^{\prime\prime}(x)-2xH_N^{\prime}+2NH_N(x),
\end{equation}
and setting $x=z_i$ after expressing the $N$-th Hermite polynomial in the form $H_N(x)=C_N\prod_{i=1}^N(x-z_i)$.

We finish the appendix by showing that the only ladder operators that map $\tilde{\mcl{H}}$ onto itself are of the form $\hat{K}^{(j,j)}\pm\hat{L}^{(j,j)}$. By definition, $\tilde{\mcl{H}}$ is spanned by vectors $|\rho\rangle$ with $\rho\in S_N$; the action of a generic ladder operator on $|\rho\rangle$ is given by
\begin{align}
(\hat{K}^{(j,l)}\pm\hat{L}^{(j,l)})|\rho\rangle=\sum_{m=1}^N \Big(&z_m\hat{J}^{(j,l)}_m|\rho\rangle\notag\\
&\pm\sum_{\substack{1\leq n\leq N\\n:n\neq m}}\sum_{\substack{(j_1,l_1)\\(j_2,l_2)}}\frac{f^{[(j,l),(j_1,l_1),(j_2,l_2)]}}{z_m-z_n}\hat{J}^{(j_1,l_1)}_m\hat{J}^{(j_2,l_2)}_n|\rho\rangle\Big).\label{eq:laddereffect}
\end{align}
If $j\neq l$, then $\sum_{m=1}^Nz_m\hat{J}^{(j,l)}_m$ takes $|\rho\rangle$ into a linear combination of 
\[|\rho(1),\ldots,\rho(k-1),j,\rho(k+1),\ldots,\rho(N)\rangle\text{ and }|\rho(1),\ldots,\rho(n-1),l,\rho(n+1),\ldots,\rho(N)\rangle,\]
since there exist $k$ and $n$ such that $\rho(k)=l$ and $\rho(n)=j$. However, if $\rho(k)$ is changed to $j$, for example, then there are two indices with the same value: the $k$-th index, which has been changed into $j$, and the $n$-th index, which is $\rho(n)=j$. Then, the resulting vectors are not in the span of $\{|\rho\rangle\}_{\rho\in S_N}$ and therefore not in $\tilde{\mcl{H}}$. Moreover, since the structure constants are antisymmetric, the operators $\hat{J}^{(j_1,l_1)}_m$ and $\hat{J}^{(j_2,l_2)}_n$ in the second term must be such that $(j_1,l_1)$ and $(j_2,l_2)\neq(j,l)$. This means that only one of the two index pairs can take the value $(l,j)$, and therefore the product $\hat{J}^{(j_1,l_1)}_m\hat{J}^{(j_2,l_2)}_n$ cannot reproduce the action of any of the $\hat{J}^{(j,l)}_m$ in the first term. Consequently, all ladder operators of the form $\hat{K}^{(j,l)}\pm\hat{L}^{(j,l)}$ with $j\neq l$ are ruled out.

It only remains to show that all ladder operators $\hat{K}^{(j,j)}\pm\hat{L}^{(j,j)}$ satisfy our requirements. Because $\hat{J}^{(j,j)}_m$ is a diagonal operator in $\mcl{H}_m$, the $m$-th index remains unchanged. Then, we must check that the second term in \eqref{eq:laddereffect} maps $|\rho\rangle$ to a linear combination of vectors in $\tilde{\mcl{H}}$. A direct calculation gives
\begin{align}
[\hat{J}^{(i,i)},\hat{J}^{(j,k)}]=-\sqrt{\frac{-1}{2i(i+1)}}\Big(\mbb{I}_{j\leq i}-\mbb{I}_{k\leq i}+i(\delta_{k,i+1}-\delta_{j,i+1})\Big)\hat{J}^{(k,j)}
\end{align}
for $j<k$. It follows that 
\[f^{[(i,i),(j,k),(k,j)]}=-\sqrt{\frac{-1}{2i(i+1)}}\Big(\mbb{I}_{j\leq i}-\mbb{I}_{k\leq i}+i(\delta_{k,i+1}-\delta_{j,i+1})\Big)\]
and that
\[\sum_{(j_1,k_1)}f^{[(i,i),(j,k),(j_1,k_1)]}=f^{[(i,i),(j,k),(k,j)]}.\]
Without loss of generality, let us assume that $m=1$ and $n=2$. Then, we have
\begin{align}
\sum_{(j_1,k_1)}\sum_{(j_2,k_2)}&f^{[(i,i),(j_1,k_1),(j_2,k_2)]}\hat{J}^{(j_1,k_1)}_1\hat{J}^{(j_2,k_2)}_2\notag\\
&=\sum_{1\leq j<k\leq N}f^{[(i,i),(j,k),(k,j)]}\Big(\hat{J}^{(j,k)}_1\hat{J}^{(k,j)}_2-\hat{J}^{(k,j)}_1\hat{J}^{(j,k)}_2\Big)\notag\\
&=\frac{\sqrt{-1}}{2}\sum_{1\leq j<k\leq N}f^{[(i,i),(j,k),(k,j)]}\Big(|k\rangle_1\langle j|_1\otimes|j\rangle_2\langle k|_2\bigotimes_{l=3}^N\hat{I}_l\notag\\
&\qquad\qquad\qquad\qquad\qquad\qquad\qquad\qquad-|j\rangle_1\langle k|_1\otimes|k\rangle_2\langle j|_2\bigotimes_{l=3}^N\hat{I}_l\Big).
\end{align}
This indicates that the second term in \eqref{eq:laddereffect} produces a linear combination of vectors in which indices are always exchanged. Consequently, if the ladder operators $\hat{K}^{(j,j)}\pm\hat{L}^{(j,j)}$ act on a vector indexed by a permutation, the resulting vectors are indexed by permutations as well, meaning that they are the only ladder operators that map $\tilde{\mcl{H}}$ into itself. Since there are only $N-1$ such raising and lowering operators, our claim is proved.


\begin{thebibliography}{99}

\bibitem{RoslerVoit98} R{\"o}sler, M., Voit, M.: Markov Processes Related with Dunkl Operators.
\emph{Adv. in Appl. Math.} \textbf{21}, (1998), 575--643.

\bibitem{RoslerVoit08} R{\"o}sler, M., Voit, M.: Dunkl theory, convolution algebras, and related Markov processes. In: Harmonic and stochastic analysis of Dunkl processes. 
\emph{Hermann}, Paris, 2008. 228 pp.

\bibitem{Dunkl89} Dunkl, C.F.: Differential-Difference Operators Associated to Reflection Groups.
\emph{Trans. Amer. Math. Soc.} \textbf{311}, (1989), 167--183.

\bibitem{Dyson62} Dyson, F. J.: A Brownian-motion model for the eigenvalues of a random matrix.
\emph{J. Math. Phys.} \textbf{3}, (1962), 1191-1198.

\bibitem{DemniBook} Demni, N.: A Guided tour in the world of radial Dunkl processes. In: Harmonic and stochastic analysis of Dunkl processes. 
\emph{Hermann}, Paris, 2008. 228pp.

\bibitem{GallardoYor05} Gallardo, L., Yor, M.: Some new examples of Markov processes which enjoy the time-inversion property.
\emph{Probab. Theory Related Fields} \textbf{132}, (2005), 150--162.

\bibitem{GallardoYor06A} Gallardo, L., and Yor, M.: Some remarkable properties of the Dunkl martingales.
\emph{In memoriam: Paul-Andr{\'e} Meyer: S{\'e}minaire de Probabilit{\'e}s XXXIX}, 337--356, Lecture Notes in Math., 1874, \emph{Springer}, Berlin, 2006.

\bibitem{GallardoYor06B} Gallardo, L., and Yor, M.: A chaotic representation property of the multidimensional Dunkl processes.
\emph{Ann. Probab.} \textbf{34}, (2006), 1530-1549.

\bibitem{Chybiryakov08} Chybiryakov, O.: Skew-representations of multidimensional Dunkl Markov processes.
\emph{Ann. Inst. Henri Poincar{\'e} Probab. Stat.} \textbf{44}, (2008), 593-611.

\bibitem{GenestVinetZhedanov14} Genest, V. X., Vinet, L., Zhedanov, A.: The Bannai-Ito polynomials as Racah coefficients of the $sl_{-1}(2)$ algebra.
\emph{Proc. Amer. Math. Soc.} \textbf{142}, (2014), 1545-1560.

\bibitem{deBieGenestVinet16} De Bie, H., Genest, V. X., Vinet, L.: A Dirac-Dunkl Equation on $S^2$ and the Bannai-Ito Algebra.
\emph{Comm. Math. Phys.} \textbf{344}, (2016), 447-464.

\bibitem{BakerForrester} Baker, T., Forrester, P.J.: The Calogero-Sutherland Model and Generalized Classical Polynomials.
\emph{Commun. Math. Phys.} \textbf{188}, (1997), 175-216.

\bibitem{Khastgir} Khastgir, S.P., Pocklington, A.J., Sasaki, R.: Quantum Calogero-Moser models: integrability for all root systems.
\emph{J. Phys. A: Math. Gen.} \textbf{33}, (2000), 9033-9064.

\bibitem{EtingofGinzburg} Etingof, P., Ginzburg, V.: Symplectic reflection algebras, Calogero-Moser space, and deformed Harish-Chandra homomorphism.
\emph{Invent. Math.} \textbf{147}, (2002), 243-348.

\bibitem{Bru} Bru, M.-F.: Wishart processes.
\emph{J. Theor. Probab.} \textbf{4}, (1991), 725-751.

\bibitem{KonigOConnell} K{\"o}nig, W., O'Connell, N.: Eigenvalues of the Laguerre process as non-colliding squared Bessel processes.
\emph{Electron. Commun. Probab.} \textbf{6}, (2001), pp. 107-114.

\bibitem{Forrester} Forrester, P. J.: Log-Gases and Random Matrices. London Mathematical Society Monographs, Vol. 34.
\emph{Princeton University Press}, Princeton, 2010. xiv+791pp.

\bibitem{Polychronakos92} Polychronakos, A. P.: Exchange Operator Formalism for Integrable Systems of Particles.
\emph{Phys. Rev. Lett.} \textbf{69}, (1992), 703-705.

\bibitem{Frahm93} Frahm, H.: Spectrum of a spin chain with inverse square exchange.
\emph{J. Phys. A: Math. Gen.} \textbf{26}, (1993), L473-L479.

\bibitem{HikamiWadati93} Hikami, K., Wadati, M.: Infinite Symmetry of the Spin Systems with Inverse Square Interactions.
\emph{J. Phys. Soc. Jpn.} \textbf{62}, (1993), 4203-4217.

\bibitem{KatoriTanemura07} Katori, M., Tanemura, H.: Noncolliding Brownian Motion and Determinantal Processes.
\emph{J. Stat. Phys.} \textbf{129}, (2007), pp. 1233-1277.

\bibitem{Demni08} Demni, N.: First hitting time of the boundary of the Weyl chamber by radial Dunkl processes.
\emph{SIGMA Symm. Integrab. Geom. Methods Appl.} \textbf{4}, (2008), 074, 14pp.

\bibitem{AndrausMiyashita15} Andraus, S., Miyashita, S.: Two-step asymptotics of scaled Dunkl processes.
\emph{J. Math. Phys.} \textbf{56}, (2015), 103302, pp. 1-23.

\bibitem{DunklXu} Dunkl, C.F., Xu, Y.: Orthogonal Polynomials of Several Variables. Encyclopedia of Mathematics and its Applications, 81.
\emph{Cambridge University Press}, Cambridge, 2001. xvi+390pp.

\bibitem{DumitriuEdelman02} Dumitriu, I., Edelman, A.: Matrix Models for Beta Ensembles.
\emph{J. Math. Phys.} \textbf{43}, (2002), 5830--5847.

\bibitem{Dunkl91} Dunkl, C.F.: Integral kernels with reflection group invariance.
\emph{Canad. J. Math.} \textbf{43}, (1991), 1213 -- 1227.

\bibitem{Demni09} Demni, N.: Radial Dunkl processes: existence, uniqueness and hitting time.
\emph{C. R. Acad. Sci. Paris} \textbf{347}, (2009), 1125-1128.

\bibitem{Meyer67} Meyer, P. A.: Int{\'e}grales stochastiques. I, II, III, IV. \emph{S{\'e}minaire de Probabilit{\'e}s (Univ. Strasbourg, Strasbourg, 1966/67), Vol. I}, \emph{Springer}, Berlin, 1967. pp. 72-94, 95-117, 118-141, 142-162.

\bibitem{Lepingle12} L{\'e}pingle, D.: On the mean number of jumps of a Dunkl process. Preprint: \href{www.univ-orleans.fr/mapmo/membres/lepingle/sauts.pdf}{www.univ-orleans.fr/mapmo/membres/lepingle/sauts.pdf}

\bibitem{KatoriTanemura11} Katori, M., Tanemura, H.: Noncolliding Squared Bessel Processes.
\emph{J. Stat. Phys.} \textbf{142}, (2011), pp. 592-615.

\bibitem{OsadaTanemura14} Osada, H., Tanemura, H.: Infinite-dimensional stochastic differential equations and tail $\sigma$-fields.
arXiv:1412.8674.

\bibitem{ValkoVirag09} Valk{\'o}, B., Vir{\'a}g, B.: Continuum limits of random matrices and the Brownian carousel.
\emph{Invent. Math.} \textbf{177}, (2009), pp. 463-508.

\bibitem{AndrausKatoriMiyashita12} Andraus, S., Katori, M., Miyashita, S.: Interacting particles on the line and Dunkl intertwining operator of type $A$: application to the freezing regime.
\emph{J. Phys. A: Math. Theor.} \textbf{45}, (2012), 395201, pp. 1-26.

\bibitem{AndrausKatoriMiyashita14} Andraus, S., Katori, M., Miyashita, S.: Two limiting regimes of interacting Bessel processes.
\emph{J. Phys. A: Math. Theor.} \textbf{47}, (2014), 235201, pp. 1-30.

\bibitem{AndrausVoit18} Andraus, S., Voit, M.: Limit theorems for multivariate Bessel processes in the freezing regime. \emph{Stoch. Proc. Appl.} \textbf{129} (2019) 4771-4790.

\bibitem{VoitWoerner} Voit, M., W{\"o}rner, J. H. C.: Functional central limit theorems for multivariate Bessel processes in the freezing regime. arXiv:1901.08390.

\bibitem{AndrausVoit19} Andraus, S., Voit, M.: Central limit theorems for multivariate Bessel processes in the freezing regime II: The covariance matrices. \emph{J. Approx. Theor.} \textbf{246} (2019) 65-84.

\bibitem{Dunkl89B} Dunkl, C.F.: Harmonic polynomials and peak sets of reflection groups.
\emph{Geom. Ded.} \textbf{32}, (1989), pp. 157--171.

\bibitem{Fekete23} Fekete, M., \"Uber die {V}erteilung der {W}urzeln bei gewissen algebraischen {g}leichungen mit ganzzahligen {K}oeffizienten. \emph{Math. Z.} \textbf{17} (1923), pp. 228--249.

\bibitem{Szego} Szeg{\"o}, G.: Orthogonal polynomials. Colloquium Publications, Vol. 23.
\emph{American Mathematical Society}, Providence, 1939. xiii+431pp.

\end{thebibliography}
\end{document}